\documentclass{article}
\pdfoutput=1
\usepackage{arxiv}
\usepackage[utf8]{inputenc} % allow utf-8 input
\usepackage[T1]{fontenc}    % use 8-bit T1 fonts
\usepackage{amsmath,amssymb,amsthm}
\usepackage{mathtools}
\usepackage{enumitem}
\usepackage{booktabs}
\usepackage{microtype}      % microtypography
\usepackage{doi}
\usepackage{array}
\usepackage{hyperref}
\usepackage{cleveref}
\usepackage{algorithm}
\usepackage{algpseudocode}
\usepackage[most]{tcolorbox}
\usepackage{tikz}
\usetikzlibrary{arrows.meta, positioning, shapes.geometric, calc, backgrounds, fit,decorations.pathreplacing}
\usepackage{mathtools}

\definecolor{boxcolor1}{RGB}{255,254,230}	% yellow 	#FFFEE6
\definecolor{boxcolor2}{RGB}{252,227,215}	% red		#fce3d7
\definecolor{boxcolor3}{RGB}{225,240,227}	% green 	#E1F0E3
\definecolor{boxcolor4}{RGB}{227,235,246}	% blue		#E3EBF6

\newtcolorbox{definitionbox}[1][]{
	breakable,
	before skip=1.5\topskip,
	after skip=1.5\topskip,
	left skip=0pt,
	right skip=0pt,
	left=4pt,
	right=4pt,
	top=2pt,
	bottom=2pt,
	lefttitle=4pt,
	righttitle=4pt,
	toptitle=2pt,
	bottomtitle=2pt,
	sharp corners,
	boxrule=0pt,
	titlerule=.4pt,
	colback=boxcolor1,
	colbacktitle=boxcolor1,
	coltitle=black,
	colframe=darkgray,
	coltext=black,
	fonttitle=\bfseries,
	title=Definition~\thetcbcounter:,
	#1
}

\newtcolorbox{theorembox}[1][]{
	breakable,
	before skip=1.5\topskip,
	after skip=1.5\topskip,
	left skip=0pt,
	right skip=0pt,
	left=4pt,
	right=4pt,
	top=2pt,
	bottom=2pt,
	lefttitle=4pt,
	righttitle=4pt,
	toptitle=2pt,
	bottomtitle=2pt,
	sharp corners,
	boxrule=0pt,
	titlerule=.4pt,
	colback=boxcolor2,
	colbacktitle=boxcolor2,
	coltitle=black,
	colframe=darkgray,
	coltext=black,
	fonttitle=\bfseries,
	title=Theorem~\thetcbcounter:,
	#1
}

\newtcolorbox{lemmabox}[1][]{
	breakable,
	before skip=1.5\topskip,
	after skip=1.5\topskip,
	left skip=0pt,
	right skip=0pt,
	left=4pt,
	right=4pt,
	top=2pt,
	bottom=2pt,
	lefttitle=4pt,
	righttitle=4pt,
	toptitle=2pt,
	bottomtitle=2pt,
	sharp corners,
	boxrule=0pt,
	titlerule=.4pt,
	colback=boxcolor3,
	colbacktitle=boxcolor3,
	coltitle=black,
	colframe=darkgray,
	coltext=black,
	fonttitle=\bfseries,
	title=Lemma~\thetcbcounter:,
	#1
}

\newtcolorbox{propositionbox}[1][]{
	breakable,
	before skip=1.5\topskip,
	after skip=1.5\topskip,
	left skip=0pt,
	right skip=0pt,
	left=4pt,
	right=4pt,
	top=2pt,
	bottom=2pt,
	lefttitle=4pt,
	righttitle=4pt,
	toptitle=2pt,
	bottomtitle=2pt,
	sharp corners,
	boxrule=0pt,
	titlerule=.4pt,
	colback=boxcolor3,
	colbacktitle=boxcolor3,
	coltitle=black,
	colframe=darkgray,
	coltext=black,
	fonttitle=\bfseries,
	title=Proposition~\thetcbcounter:,
	#1
}

\newtcolorbox{corollarybox}[1][]{
	breakable,
	before skip=1.5\topskip,
	after skip=1.5\topskip,
	left skip=0pt,
	right skip=0pt,
	left=4pt,
	right=4pt,
	top=2pt,
	bottom=2pt,
	lefttitle=4pt,
	righttitle=4pt,
	toptitle=2pt,
	bottomtitle=2pt,
	sharp corners,
	boxrule=0pt,
	titlerule=.4pt,
	colback=boxcolor4,
	colbacktitle=boxcolor4,
	coltitle=black,
	colframe=darkgray,
	coltext=black,
	fonttitle=\bfseries,
	title=Corollary~\thetcbcounter:,
	#1
}

\newtheoremstyle{mystyle}% name
{3pt}% Space above
{3pt}% Space below
{\itshape}% Body font
{}% Indent amount
{\bfseries}% Theorem head font
{.}% Punctuation after theorem head
{.5em}% Space after theorem head
{}% Theorem head spec (can be left empty, meaning ‘normal’)

\theoremstyle{mystyle}

\newtheorem{example}{Example}[section]

\newtheorem{definition}{Definition}[section] % Ensure it is part of the section enumeration

\renewenvironment{definition}{%
	\refstepcounter{definition}% Increment the definition counter
	\begin{definitionbox}[title=Definition~\thedefinition]% Use the current counter value in the title
	}{%
	\end{definitionbox}
}

\newtheorem{theorem}{Theorem}[section] % Ensure it is part of the section enumeration

\renewenvironment{theorem}{%
	\refstepcounter{theorem}% Increment the definition counter
	\begin{theorembox}[title=Theorem~\thetheorem]% Use the current counter value in the title
	}{%
	\end{theorembox}
}

\newtheorem{lemma}{Lemma}[section] % Ensure it is part of the section enumeration

\renewenvironment{lemma}{%
	\refstepcounter{lemma}% Increment the definition counter
	\begin{lemmabox}[title=Lemma~\thelemma]% Use the current counter value in the title
	}{%
	\end{lemmabox}
}

\newtheorem{proposition}{Proposition}[section] % Ensure it is part of the section enumeration

\renewenvironment{proposition}{%
	\refstepcounter{proposition}% Increment the definition counter
	\begin{propositionbox}[title=Proposition~\theproposition]% Use the current counter value in the title
	}{%
	\end{propositionbox}
}

\newtheorem{corollary}{Corollary}[section] % Ensure it is part of the section enumeration

\renewenvironment{corollary}{%
	\refstepcounter{corollary}% Increment the definition counter
	\begin{corollarybox}[title=Corollary~\thecorollary]% Use the current counter value in the title
	}{%
	\end{corollarybox}
}

\newtheorem*{remark}{Remark}
\renewenvironment{proof}{{\noindent\bfseries Proof:}}{\qed}

\newtheoremstyle{claimstyle} % Define a new theorem style
{3pt} % Space above
{3pt} % Space below
{} % Body font (default upright, non-italic)
{} % Indent amount
{\bfseries} % Head font (bold)
{.} % Punctuation after head
{0.5em} % Space after head
{} % Head spec

\theoremstyle{claimstyle}
\newtheoremstyle{postulatestyle} % Define a new theorem style
{3pt} % Space above
{3pt} % Space below
{} % Body font (default upright, non-italic)
{} % Indent amount
{\bfseries} % Head font (bold)
{.} % Punctuation after head
{0.5em} % Space after head
{} % Head spec

\theoremstyle{postulatestyle}
\newtheoremstyle{questionstyle} % Define a new theorem style
{3pt} % Space above
{3pt} % Space below
{} % Body font (default upright, non-italic)
{} % Indent amount
{\bfseries} % Head font (bold)
{.} % Punctuation after head
{0.5em} % Space after head
{} % Head spec

\theoremstyle{questionstyle}
\newtheorem{question}{Question} % Unnumbered claim environment

\newtheoremstyle{assumptionstyle} % Define a new theorem style
{3pt} % Space above
{3pt} % Space below
{} % Body font (default upright, non-italic)
{} % Indent amount
{\bfseries} % Head font (bold)
{.} % Punctuation after head
{0.5em} % Space after head
{} % Head spec

\theoremstyle{assumptionstyle}
\newtheorem{assumption}{Assumption} % Unnumbered claim environment

\newtheoremstyle{conjecturestyle} % Define a new theorem style
{3pt} % Space above
{3pt} % Space below
{} % Body font (default upright, non-italic)
{} % Indent amount
{\bfseries} % Head font (bold)
{.} % Punctuation after head
{0.5em} % Space after head
{} % Head spec

\theoremstyle{conjecturestyle}
\newtheorem{conjecture}{Conjecture} % Unnumbered claim environment

\newtheoremstyle{critiquestyle} % Define a new theorem style
{3pt} % Space above
{3pt} % Space below
{\itshape} % Body font (default upright, non-italic)
{} % Indent amount
{\bfseries} % Head font (bold)
{.} % Punctuation after head
{0.5em} % Space after head
{} % Head spec

\theoremstyle{critiquestyle}
\newtheoremstyle{responsestyle} % Define a new theorem style
{3pt} % Space above
{3pt} % Space below
{\itshape} % Body font (default upright, non-italic)
{} % Indent amount
{\bfseries} % Head font (bold)
{.} % Punctuation after head
{0.5em} % Space after head
{} % Head spec

\theoremstyle{responsestyle}
\newcommand{\Actions}{\mathcal{A}}

\newcommand{\Features}{\mathcal{F}}
\newcommand{\Params}{\Theta}
\newcommand{\owners}{O}
\newcommand{\status}{\sigma}
\newcommand{\locked}{L}
\newcommand{\primturn}{T}
\newcommand{\valmat}{\mathbf{V}}
\newcommand{\appear}{\mathbf{a}}

\DeclareMathOperator{\Range}{Range}

\DeclareMathOperator{\clip}{clip}

\newcommand{\PI}{\texttt{PI}}
\newcommand{\SC}{\texttt{SC}}
\newcommand{\AD}{\texttt{AD}}
\newcommand{\BS}{\texttt{BS}}

\title{Formal specification and behavioral simulation of the holiday gift exchange game}

\author{ \href{https://orcid.org/0009-0004-7957-1806}{\includegraphics[scale=0.06]{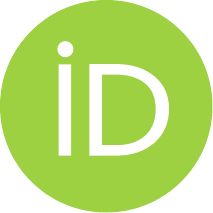}\hspace{1mm}Daniel Quigley} \\
	Center for Possible Minds\\
	Indiana University Bloomington\\
	Bloomington, IN 47408 \\
	\texttt{dgquigle@iu.edu} \\
}

\renewcommand{\headeright}{Gift exchange game}
\renewcommand{\undertitle}{Preprint}
\renewcommand{\shorttitle}{Gift exchange game}
\colorlet{linecol}{cyan!90!blue!90!black}
\colorlet{fillcol}{cyan!60!blue!80!black!40}
\hypersetup{
pdftitle={Formal specification and behavioral simulation of the holiday gift exchange game},
pdfsubject={gift exchange game, formalization, simulation},
pdfauthor={Daniel Quigley},
pdfkeywords={gift exchange game, formalization, simulation},
}

\begin{document}
\maketitle

\begin{abstract}
	The holiday gift exchange game is a familiar social institution with nontrivial strategic structure. We provide a formal treatment of the game's mechanics, defining the state space, action sets, and the recursive structure of stealing chains; we prove termination and derive an algorithm for counting distinct game trajectories, which grow far faster than the space of possible final allocations. Beyond the base mechanics, we introduce a decorated model incorporating partial information, social costs, and adaptive strategies grounded in discrete choice theory and the frustration-aggression literature. A full factorial simulation of 240,000 games yields three findings of note: implicit social costs are the dominant regulator of aggression, reducing stealing by 27--48\% and outweighing both uncertainty and strategic sophistication; partial information, contrary to expectation, slightly increases stealing through asymmetric uncertainty; correlated valuations amplify every behavioral effect, so that consensus about gift quality, rather than the features themselves, is what intensifies competition. The first-player advantage is robust across all conditions.

\end{abstract}


\section{Introduction}

The gift exchange game is played at gatherings worldwide under various names: White Elephant; Yankee Swap; Dirty Santa. Despite its ubiquity, the game has received little formal treatment. The rules are simple enough to explain in minutes, yet generate strategic complexity, social friction, and, when things go well, entertainment.

The basic structure: $n$ players each contribute a wrapped gift. Players act in sequence. On one's turn, one may either open a wrapped gift or steal an already-opened gift from another player. Stealing displaces the victim, who must then act (open or steal, but not steal back the just-taken gift). This creates chains of thefts terminating when someone opens. After all primary turns complete, the first player receives a final swap opportunity, compensating for the initial absence of steal targets.

This paper pursues the following objectives. First, we provide a complete formal specification of the game mechanics, appropriately so  to admit proofs (e.g., termination, end-state bijectivity) and flexible enough to parameterize common rule variants (per-round and lifetime steal limits). Second, we construct a behavioral model incorporating features absent from purely strategic analysis: uncertainty about gift values; social costs of stealing; emotional dynamics (frustration, satisfaction); appearance-based selection biases. Third, we simulate the decorated game systematically to isolate feature effects and their interactions.

This paper is in direct conversation with \cite{apagodu2017,applegate2009giftexchangeproblem}, and we answer the same question: determining the number of ways that a game can be played out for given values of number of players and number of steals. Where we differ is in our treatment of the formalization of the base game rules, and in our accounting of variations to the base game; in the former, we provide an exhaustive formal description of the game and problem space, and in the latter, we account explicitly for trajectories of played out games sensitive to differences in final states. Furthermore, we assemble computational proxies of human actors' cognitive aspects in competition games. The crescendo of our work is an analysis of simulations across multiple game trajectories sensitive to player motiviations in addition to calculating possible number of games.

Our behavioral model draws on discrete choice theory \cite{manski2001daniel,train2009,mcfadden2001economic}, social preference models \cite{fehr2014social, charness2002understanding}, and the literature on emotions in economic decisions \cite{loewenstein2000emotions,lerner2015emotion,pfister2008multiplicity}. The frustration dynamics we imnplement here adapt the frustration-aggression hypothesis from social psychology \cite{berkowitz1989frustration,kruglanski2023frustration}.

Gift exchanges appear occasionally in recreational mathematics and puzzle columns, but we are unaware of any prior formal treatment as given here. The game is mentioned in popular accounts of holiday traditions and occasionally surfaces in discussions of ``fun'' mechanism design, but lacking in mathematical content.

This paper is arranged into three general parts. The first material contains Sections~\ref{sec:basegame} and \ref{sec:decoratedgame} is the formalization proper, and defines the base game and the decorated game: players, gifts, actions, state transitions, stealing chains, proves termination, and introduces parameterized stealing limits in the former, while the latter provides extentions to the base game. Section~\ref{sec:experimentframework} walks through the experiment design for the simulation of the gift exchange game, implementing the materials from the formalization, while Section~\ref{sec:experimenrresults} analyzes the results of the simulation experiment. Section~\ref{sec:trajectory} works through the combinatorics of how many games are possible given our suite of rules on stealing. Section~\ref{sec:future} discusses future directions, including coalition formation and mechanism design questions. Section~\ref{sec:conclusion} concludes. The appendices provide additional material: notation reference, reports of complexities, and tabulated results.

\subsection{Disclaimer: a note on scope}

Let us state forthwith our specific study: the variant of the exchange game in which stealing displaces the victim (who then takes a turn) rather than triggering a swap. Other variants exist and could be analyzed using the same framework.

We do not claim that players consciously solve optimization problems, or that the behavioral parameters we estimate are universal or truly modeling human behavior. The model is a tool for  structured reasoning about the game's dynamics, and we invite a larger scale psychological theory in future work. Our goal is to illuminate the strategic and social structure of a familiar institution.

\subsection{Overview of gift exchange game}

All players\footnote{Players may be composite (several people as an ensemble player together).} bring a wrapped gift\footnote{Gifts may be composite (several items as an ensemble together). Additional limiters on what kinds of gifts are used may be imposed (themed gifts, gag gifts, all value under some amount, etc.).} and places in a central location. Players determine turn order in some appropriate fashion. When it is a player's turn, that player may either:

\begin{itemize}
	\item open a new gift;
	\item steal an already-opened gift from another player.
\end{itemize}

If a player chooses to open a new gift, then that gift is shown to all players, and that player owns the gift for now. Alternatively, if that player chooses to steal an already-opened gift from another player, then that other player must hand over their gift. The person from whom the first player stole the gift now gets to go again immediately. They can either:

\begin{itemize}
	\item open a new gift;
	\item steal an already-opened gift from another player, but may not steal back the gift just stolen.
\end{itemize}

This can set off a chain reaction: if player A steals from player B, player B might steal from player C, and so on; the chain ends when someone decides to open a new gift instead of stealing.

When a gift gets stolen, it is locked for the rest of that chain. This prevents endless back-and-forth: if player A steals player B's gift, then player B cannot immediately steal it back from player A, and must instead steal from someone else or open a new gift. Once the chain ends (someone opens a gift), all gifts become unlocked, and are considered fair game again in future turns.

After everyone has taken their turn, the game is not quite over: the player who went first at the beginning of the game gets one last chance to swap; since they went first (when there was nothing to steal), they get to survey all the opened gifts and swap with anyone if they want to. This is an optional action: they can instead choose to keep what they have. This final swap is a strict trade, and so does not trigger a chain, and the game ends.

Some rules of the base game may be amended for a different play experience:

\begin{itemize}
	\item in the base game, a gift can only be stolen once per round (per turn); an alternative version is to allow gifts to be stolen multiple times in the same round;
	\item in the base game, there is no limit to the number of times a gift can be stolen across the whole game; alternatively, after a gift has been stolen some number of times total (across all rounds), it is locked and cannot be stolen again at all, and the current holder keeps it safe;
\end{itemize}

The gift exchange game is supposed to be fun, a little chaotic, and maybe slightly frustrating in a good-natured way; please do no take it too seriously and be \emph{that} guy at the gaming table.

\subsection{Some history}

Much of this material follows from \cite{brown2014whiteelephantyankee,brown2014whiteelephantidoms}; please refer to them for greater detail and additional sources.

The tradition of exchanging gifts in a game-like format has multiple regional names across the United States, including ``White Elephant'', ``Yankee Swap'', and ``Dirty Santa'', along with more colorful variations such as ``Thieving Elves'', and ``Cutthroat Christmas''. While the modern practice centers on exchanging either practical items or deliberately impractical gifts in an entertaining format, its origins can be traced to early American swap parties’ documented as far back as 1901 in \textit{The Hartford Herald}.

The term ``Yankee Swap'', appearing in Walt Whitman’s 1855 preface to \textit{Leaves of Grass}, deliberately calls out the reputation of Yankees as shrewd traders. The ``White Elephant'' nomenclature emerged later, first documented in 1907, deriving from a supposed Southeast Asian practice of kings gifting rare white elephants to courtiers, though historians note this origin story lacks historical verification, as white elephants were highly valued in Thai culture rather than viewed as burdens.

The earliest description of these gift exchanges appears in 1901, describing a swap party: participants brought four or five little neatly wrapped and tied bundles, with the directive that ‘the more misleading in shape as to contents the better’. The practice quickly gained popularity, with detailed instructions appearing in various publications including Clara E. Laughlin’s \textit{The Complete Hostess} (1906) and Madame Merri’s \textit{The Art of Entertaining for All Occasions} (1913).

\section{Base game structure}\label{sec:basegame}

Throughout, we adopt the following notational conventions:

\begin{itemize}
	\item $[n] := \{1, 2, \ldots, n\}$ index set for $n$ elements;
	\item $\bot$ distinguished ``null'' or ``empty'' element (\emph{not} to be confused with the empty set $\emptyset$);
	\item $\clip_{[a,b]}(x) := \max(a, \min(b, x))$ clips $x$ to the interval $[a,b]$;
	\item $\mathcal{P}(S)$ the power set of $S$;
	\item $|S|$ the cardinality of set $S$.
\end{itemize}

Let us walk through the components of the base game, and provision a technical foramalization throughout.

\subsection{Players and gifts}

\begin{definition}[Player set]
	Let $P = \{P_1, P_2, \ldots, P_n\}$ be the set of $n$ players, indexed by their position in the turn order; player $P_1$ acts first in round 1, $P_2$ in round 2, and so forth.
\end{definition}

The player $P_n$ is a label for an actor in the game. Though singular, the ``player'' may well be more than one individual; the constraints remain the same, in that $P_n$ provisions a gift index labeled $n$.

\begin{definition}[Gift set]
	Let $G = \{G_1, G_2, \ldots, G_n\}$ be the set of $n$ gifts; assume $|P| = |G| = n$, such that each player receives exactly one gift by game's end.
\end{definition}

Note that what constitutes a ``gift'' could be a singular object, or plural for that same index. We define the composite of such gifts that share the index $n$.

\begin{definition}[Composite gift]\label{def:composite}
	A \emph{composite gift} $G_j$ consists of $k \geq 1$ subcomponents:
	\[
	G_j = (g_1, g_2, \ldots, g_k)
	\]
	where each $g_i$ is an individual item. 
\end{definition}

In this way, composite gifts are atomic for game purposes: they cannot be divided; all components are revealed simultaneously upon opening. When $k = 1$, we simply have a standard gift $G_j = (g_1)$.

\begin{example}[Composite gift]
	Consider a composite gift $G_7 = (g_1, g_2, g_3)$ containing:
	\begin{itemize}
		\item $g_1$: a gardening book;
		\item $g_2$: a recorder;
		\item $g_3$: a watercolor painting kit.
	\end{itemize}
	This gift is treated as indivisible; when stolen, $G_7$ transfers as a unit, and when opened, all three items are revealed together.
\end{example}

\subsection{Gift status}

Gifts are in one of two states: open and closed, which we cast here as functions.

\begin{definition}[Status function]
	The \emph{status function} $\status: G \rightarrow \{\textsc{wrapped}, \textsc{opened}\}$ tracks whether each gift has been opened:
	\[
	\status(G_j) = \begin{cases}
		\textsc{wrapped} & \text{if } G_j \text{ has not been opened} \\
		\textsc{opened} & \text{if } G_j \text{ has been opened}
	\end{cases}
	\]
\end{definition}

Initially, $\status(G_j) = \textsc{wrapped}$ for all $G_j \in G$. Note that ``open'' and ``closed'' refer to states of the gift for the purpose of the game; whether the gift is literally closed or not is not necessary.

\begin{definition}[Wrapped and opened gift sets]
	At any point in the game:
	\begin{align}
		W &:= \{G_j \in G : \status(G_j) = \textsc{wrapped}\} \\
		U &:= \{G_j \in G : \status(G_j) = \textsc{opened}\}
	\end{align}
\end{definition}

Note that  gift sets behave as: $W \cup U = G$ and $W \cap U = \emptyset$.

\subsection{Ownership}

Let us now keep track of the possession of game objects.

\begin{definition}[Ownership function]\label{def:ownership}
	The \emph{ownership function} $\owners: P \rightarrow G \cup \{\bot\}$ assigns each player to either a gift or to the null element $\bot$, indicating no current ownership:
	\[
	\owners(P_i) = \begin{cases}
		G_j & \text{if } P_i \text{ currently owns } G_j \\
		\bot & \text{if } P_i \text{ owns no gift}
	\end{cases}
	\]
	
\end{definition}

Initially, since non player owns a gift, we have that $\owners(P_i) = \bot$ for all $P_i \in P$.

\begin{remark}
	A note on our choice of notation. We use $\bot$ rather than $\emptyset$ to avoid confusion: $\bot$ is a distinguished element in the codomain; $\emptyset$ the empty set. We need this distinction so we may define injectivity and surjectivity of $\owners$.
\end{remark}

For a gift $G_j \in G$, we define:
\[
\owners^{-1}(G_j) := \{P_i \in P : \owners(P_i) = G_j\}
\]

At any valid game state, $|\owners^{-1}(G_j)| \leq 1$ for all $G_j$.

\begin{definition}[Currently owned gifts]
	The set of currently owned gifts is:
	\[
	\Range(\owners) := \{G_j \in G : \exists P_i \in P \text{ such that } \owners(P_i) = G_j\}
	\]
\end{definition}

\subsection{Game state}

The game proceeds sequentially in discrete steps. At each step, we may check the state of gifts, players, and ownership.

\begin{definition}[Game state]\label{def:gamestate}
	The \emph{game state} at round $k$ is the tuple:
	\[
	S_k = \langle \owners, \status, \locked_k, k, \primturn \rangle
	\]
	where:
	\begin{itemize}
		\item $\owners: P \rightarrow G \cup \{\bot\}$ current ownership function;
		\item $\status: G \rightarrow \{\textsc{wrapped}, \textsc{opened}\}$ current status function;
		\item $\locked_k \subseteq G$ set of \emph{locked} gifts (cannot be stolen in the current round);
		\item $k \in [n]$ current round number;
		\item $\primturn \subseteq P$ set of players who have taken their \emph{primary turn} in the current round.
	\end{itemize}
\end{definition}

\begin{definition}[Initial state]
	The initial game state is the special state in which
	\[
	S_0 = \langle \owners_0, \status_0, \emptyset, 1, \emptyset \rangle
	\]
	where $\owners_0(P_i) = \bot$ for all $P_i \in P$ and $\status_0(G_j) = \textsc{wrapped}$ for all $G_j \in G$.
\end{definition}

\subsection{Actions}

We now inventory the actions a player may take: opening a gift and stealing a gift. The result is essentially a discrete representation in the recursive branching shown in Figure~\ref{fig:gametree} that occurs during a single player's turn. 

When player $P_k$'s turn begins, they face a binary choice represented by the 
two branches descending from the root node: the left branch leads to ``Open'', in which the player unwraps a new gift, and their turn ends immediately, and the chain 
length is $\ell = 0$, which guarentees that guarantees every chain eventually terminates.

The right branch leads to ``Steal,'' which triggers the chain mechanism highlighted by the dashed yellow boundary. When $P_k$ steals from victim $P_m$, $P_m$ becomes ``displaced'' and must now make their own decision. We show this with an arrow labeled ``displaced'' connecting the Steal action to the victim's decision node.

The displaced player $P_m$ faces the same binary choice, subject to \emph{chain-locking}: the gift that $P_k$ just stole is temporarily frozen, insofar as $P_m$ cannot steal it back; this prevents infinite cycles and ensures the chain progresses toward termination. If $P_m$ opens, the chain ends with length $\ell = 1$; if $P_m$ steals from a third player $P_j$, then that player becomes displaced, and the process continues. Each additional steal increments $\ell$ by one. 

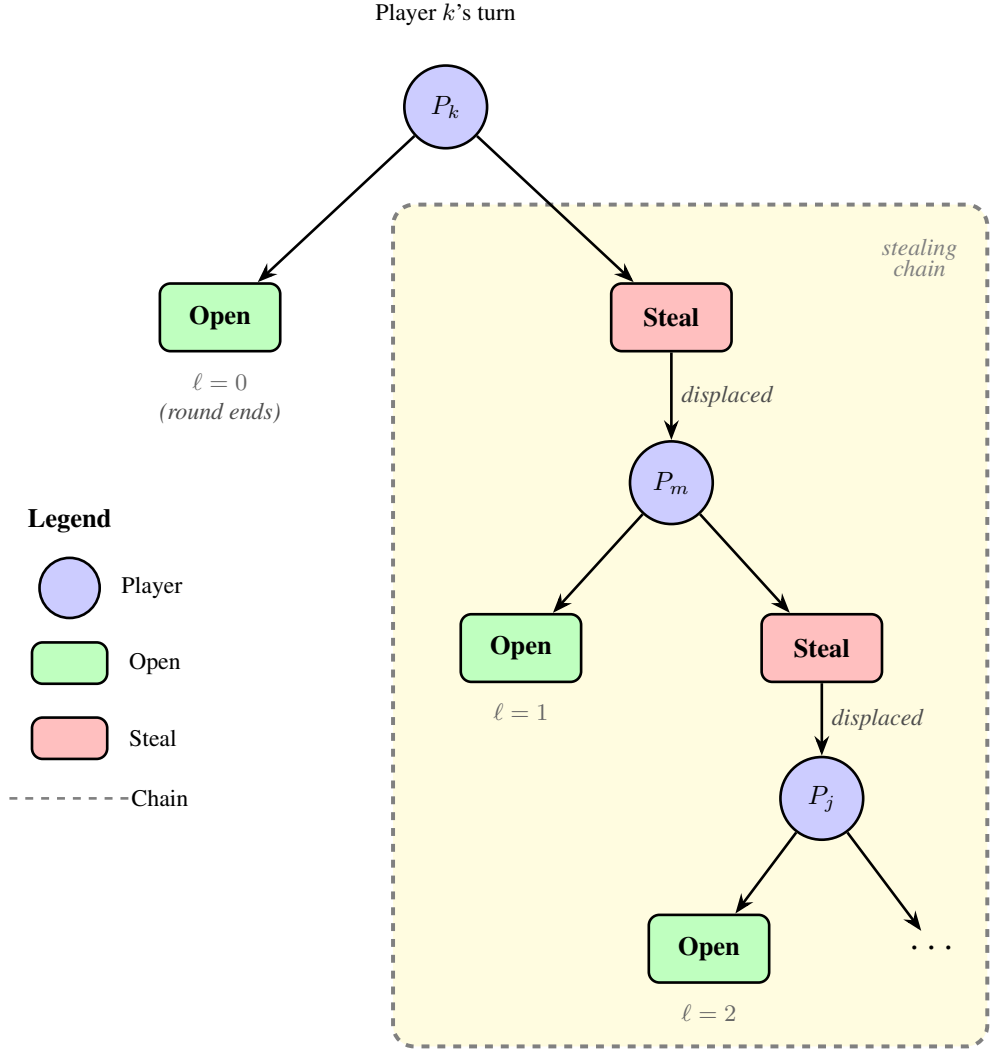
\begin{figure}[H]
	\centering
	\begin{tikzpicture}[
		% Node styles
		player/.style={
			circle,
			draw=black,
			line width=1pt,
			fill=blue!20,
			minimum size=1.1cm,
			font=\normalsize
		},
		open/.style={
			rectangle,
			draw=black,
			line width=1pt,
			fill=green!25,
			rounded corners=4pt,
			minimum width=1.6cm,
			minimum height=0.9cm,
			font=\normalsize\bfseries
		},
		steal/.style={
			rectangle,
			draw=black,
			line width=1pt,
			fill=red!25,
			rounded corners=4pt,
			minimum width=1.6cm,
			minimum height=0.9cm,
			font=\normalsize\bfseries
		},
		% Arrow style
		arr/.style={
			->,
			>=Stealth,
			line width=1pt
		},
		% Label style
		lbl/.style={
			font=\small\itshape,
			text=black!70
		},
		chainlen/.style={
			font=\small\itshape,
			text=black!60
		}
		]

		% CHAIN BOUNDARY (drawn first, behind everything)
		
		\begin{scope}[on background layer]
			\fill[yellow!15, rounded corners=8pt] 
			(2.8, -1.3) rectangle (10.7, -12.5);
			\draw[gray, dashed, line width=1.5pt, rounded corners=8pt] 
			(2.8, -1.3) rectangle (10.7, -12.5);
		\end{scope}
		
		% Chain label
		\node[font=\small\itshape, text=gray, align=center] at (9.8, -2.0) {stealing\\[-2pt]chain};

		% LEVEL 0: Primary player
		
		\node[player] (Pk) at (3.5, 0) {$P_k$};
		\node[font=\small, above=0.4cm of Pk] {Player $k$'s turn};

		% LEVEL 1: Open or Steal choice
		
		\node[open] (open1) at (0.5, -2.8) {Open};
		\node[steal] (steal1) at (6.5, -2.8) {Steal};
		
		% Arrows from Pk
		\draw[arr] (Pk) -- (open1);
		\draw[arr] (Pk) -- (steal1);
		
		% Chain length annotation for immediate open
		\node[chainlen, below=0.15cm of open1] {$\ell = 0$};
		\node[lbl, below=0.55cm of open1] {(round ends)};

		% LEVEL 2: First victim
		
		\node[player] (Pm) at (6.5, -5.0) {$P_m$};
		\draw[arr] (steal1) -- node[right, lbl] {displaced} (Pm);

		% LEVEL 3: Victim's choice
		
		\node[open] (open2) at (4.5, -7.2) {Open};
		\node[steal] (steal2) at (8.5, -7.2) {Steal};
		
		\draw[arr] (Pm) -- (open2);
		\draw[arr] (Pm) -- (steal2);
		
		% Chain length annotation
		\node[chainlen, below=0.15cm of open2] {$\ell = 1$};

		% LEVEL 4: Second victim
		
		\node[player] (Pj) at (8.5, -9.2) {$P_j$};
		\draw[arr] (steal2) -- node[right, lbl] {displaced} (Pj);

		% LEVEL 5: Second victim's choice
		
		\node[open] (open3) at (7, -11.2) {Open};
		\node[font=\Large\bfseries] (dots) at (10, -11.2) {$\cdots$};
		
		\draw[arr] (Pj) -- (open3);
		\draw[arr] (Pj) -- (dots);
		
		% Chain length annotation
		\node[chainlen, below=0.15cm of open3] {$\ell = 2$};

		% LEGEND
		
		\node[font=\normalsize\bfseries] at (-1.5, -5.5) {Legend};
		
		\node[player, minimum size=0.8cm] (leg1) at (-1.5, -6.4) {};
		\node[right=0.15cm of leg1, font=\small] {Player};
		
		\node[open, minimum width=1cm, minimum height=0.55cm] (leg2) at (-1.5, -7.4) {};
		\node[right=0.15cm of leg2, font=\small] {Open};
		
		\node[steal, minimum width=1cm, minimum height=0.55cm] (leg3) at (-1.5, -8.4) {};
		\node[right=0.15cm of leg3, font=\small] {Steal};
		
		\draw[gray, dashed, line width=1pt] (-2.3, -9.2) -- (-0.7, -9.2);
		\node[right=-0.1cm, font=\small] at (-0.7, -9.2) {Chain};
		
	\end{tikzpicture}
	\caption{Branching representation of evolving game state. Every player decides: open (escape the chain) or steal (extend it). The dashed yellow region demarcates the ``stealing chain'', everything that unfolds after an initial theft. Every decision node has exactly two children, which obfuscates the combinatorial explosion that results from recursion. When $P_k$ steals from $P_m$, the ``displaced'' label on the connecting arrow shows that $P_m$ inherits the decision, and the  victim becomes the new actor, facing the same binary choice. The chain length $\ell$ labels at each ``Open'' terminal show the depth reached before escape.}
	\label{fig:gametree}
\end{figure}

\begin{definition}[Action set]
	The \emph{action set} available to player $P_i$ at state $S_k$ is:
	\[
	\Actions(P_i, S_k) = \{\textsc{Open}(G_j) : G_j \in W\} \cup \{\textsc{Steal}(P_m) : P_m \in \mathcal{T}(P_i, S_k)\}
	\]
	where $W$ is the set of wrapped gifts and $\mathcal{T}(P_i, S_k)$ is the set of valid steal targets:
	\[
	\mathcal{T}(P_i, S_k) := \{P_m \in P : P_m \neq P_i,\, \owners(P_m) \neq \bot,\, \owners(P_m) \notin \locked_k\}
	\]
\end{definition}

\begin{definition}[Open action]\label{def:open}
	The action $\textsc{Open}(G_j)$ by player $P_i$, where $\status(G_j) = \textsc{wrapped}$, updates the state as follows:
	\begin{enumerate}
		\item $\owners(P_i) \leftarrow G_j$
		\item $\status(G_j) \leftarrow \textsc{opened}$
	\end{enumerate}
\end{definition}

Importantly, opening terminates the current action sequence, and the player $P_{n+1}$ takes their turn.

\begin{definition}[Steal action]\label{def:steal}
	The action $\textsc{Steal}(P_m)$ by player $P_i$, where $P_m \in \mathcal{T}(P_i, S_k)$, updates the state as follows; let $G_j = \owners(P_m)$:
	\begin{enumerate}
		\item $\owners(P_i) \leftarrow G_j$;
		\item $\owners(P_m) \leftarrow \bot$;
		\item $\locked_k \leftarrow \locked_k \cup \{G_j\}$;
		\item $P_m$ becomes \emph{displaced}, and must immediately take an action.
	\end{enumerate}
\end{definition}

\subsection{Turns and rounds}

Openings and stealings all constitute available actions in a turn.

\begin{definition}[Primary turn]
	A \emph{primary turn} is the initial action taken by a player during their designated round. In round $R_k$, player $P_k$ takes their primary turn.
\end{definition}

Upon completion, $\primturn \leftarrow \primturn \cup \{P_k\}$.

\begin{definition}[Turn]
	A \emph{turn} (now note without qualifier) refers to any action taken by a player, including:
	\begin{enumerate}
		\item a primary turn;
		\item an action taken after being displaced during a stealing chain.
	\end{enumerate}
\end{definition}

\begin{definition}[Round]
	\emph{Round} $R_k$ consists of all actions initiated by player $P_k$'s primary turn, including any subsequent stealing chain; a round $R_k$ concludes when the final displaced player opens a gift (or no legal steal exists).
\end{definition}
At the conclusion of round $R_k$:
\begin{enumerate}
	\item $\locked_k \leftarrow \emptyset$ (all gifts unlocked);
	\item $\primturn \leftarrow \emptyset$ (reset primary turn tracking);
	\item $k \leftarrow k + 1$.
\end{enumerate}

\subsection{Stealing chains}

When a player steals, the displaced victim must respond, potentially triggering a cascade of steals. We formalize this as a \emph{stealing chain}, and specify the constraints that ensure termination. We must, therefore, spend some time walking through the elements of the stealing action and subsequent chain. Let us first track the status of stealing.

The game includes two configurable limit parameters:
\begin{itemize}
	\item $\ell_{\text{round}} \in \mathbb{Z}_{\geq 0}$ is the maximum number of times a gift can be 
	stolen per round (0 denotes unlimited);
	\item $\ell_{\text{total}} \in \mathbb{Z}_{\geq 0}$ is the maximum number of times a gift can be 
	stolen across the entire game (0 denotes unlimited)
\end{itemize}

The standard gift exchange rules correspond to $(\ell_{\text{round}}, \ell_{\text{total}}) = (1, 0)$; this is entirely a changeable parameter, and has consequences for the number of possible games. We assume that the base game has as stated above: $(\ell_{\text{round}}, \ell_{\text{total}}) = (1, 0)$.

The game state, then, is extended with three tracking components:
\begin{itemize}
	\item $C_k \subseteq G$ are \emph{chain-locked gifts}: gifts stolen in the current stealing chain (cleared when chain terminates);
	\item $R_k: G \to \mathbb{Z}_{\geq 0}$: round steal counts of the form $R_k(G_j)$, 
	the number of times $G_j$ has been stolen in round $k$;
	\item $T: G \to \mathbb{Z}_{\geq 0}$ are total steal counts of the form $T(G_j)$, 
	the number of times $G_j$ has been stolen across all rounds
\end{itemize}

Initially, the steal tracking state is $C_k = \emptyset$, $R_k(G_j) = 0$, and $T(G_j) = 0$ for all $G_j \in G$.

\begin{definition}[Stealability predicate]
	Gift $G_j$ is \emph{stealable}, written $\textsc{Stealable}(G_j)$, if and only if  all of the following hold:
	\begin{enumerate}
		\item $G_j \notin C_k$ \hfill  chain constraint (mandatory)
		\item $\ell_{\text{round}} = 0$ or $R_k(G_j) < \ell_{\text{round}}$ \hfill Round constraint
		\item $\ell_{\text{total}} = 0$ or $T(G_j) < \ell_{\text{total}}$ \hfill total constraint
	\end{enumerate}
\end{definition}

\begin{remark}[Chain locking is mandatory]
	Chain locking ensures each gift changes  hands at most once per stealing chain; as such, the chain constraint is \emph{not} configurable, and is required to prevent infinite loops. Without it, player $P_A$ could steal from $P_B$, then $P_B$ could immediately steal back from $P_A$, and so on indefinitely! 
\end{remark}

\begin{definition}[Valid steal targets]\label{def:valid-targets}
	The set of valid steal targets for player $P_i$ at state $S_k$ is:
	\[
	\mathcal{T}(P_i, S_k) := \{P_m \in P : P_m \neq P_i,\, O(P_m) \neq \bot,\, \textsc{Stealable}(O(P_m))\}
	\]
\end{definition}

Now with our states set, let us walk through stealing chains proper.

\begin{definition}[Stealing chain]\label{def:chain}
	A \emph{stealing chain} is a sequence $\mathcal{C} = (P_{i_1}, P_{i_2}, \ldots, P_{i_m})$ 
	of players where:
	\begin{enumerate}
		\item $P_{i_1}$ initiates the chain with $\textsc{Steal}(P_{i_2})$;
		\item for each $j \in \{2, \ldots, m-1\}$, player $P_{i_j}$ was displaced by 
		$P_{i_{j-1}}$ and chose $\textsc{Steal}(P_{i_{j+1}})$;
		\item $P_{i_m}$ terminates the chain by executing $\textsc{Open}(G)$ for some $G \in W$.
	\end{enumerate}
	
\end{definition}

The \emph{length} of chain $\mathcal{C}$ is $|\mathcal{C}| - 1 = m - 1$, counting the number of steals. Let us now update our stealing definition from \ref{def:steal} to below in \ref{def:steal-action} to account for a chain.

\begin{definition}[Steal action (with chain)]\label{def:steal-action}
	The action $\textsc{Steal}(P_m)$ by player $P_i$, where $P_m \in \mathcal{T}(P_i, S_k)$, updates the state as follows; let $G_j = O(P_m)$:
	\begin{enumerate}
		\item $O(P_i) \leftarrow G_j$ \hfill transfer ownership
		\item $O(P_m) \leftarrow \bot$ \hfill victim loses gift
		\item $C_k \leftarrow C_k \cup \{G_j\}$ \hfill chain-lock the gift
		\item $R_k(G_j) \leftarrow R_k(G_j) + 1$ \hfill increment round count
		\item $T(G_j) \leftarrow T(G_j) + 1$ \hfill increment total count
		\item $P_m$ becomes \emph{displaced} and must immediately take an action.
	\end{enumerate}
\end{definition}

A stealing chain terminates when a displaced player executes $\textsc{Open}(G)$; upon termination:
\[
C_k \leftarrow \emptyset
\]

The round counts $R_k$ and total counts $T$ are \emph{not} reset at chain termination.

\begin{lemma}[Chain termination bound]\label{lem:termination}
	Every stealing chain terminates in at most $n - 1$ steals, regardless of the values of $\ell_{\text{round}}$ and $\ell_{\text{total}}$.
\end{lemma}

\begin{proof}
	Each steal adds exactly one gift to $C_k$ (the chain-locked set). Since chain locking is mandatory and cannot be disabled, the set $C_k$ grows monotonically during a chain. At most $n - 1$ gifts can be owned at any point during the game (since at least one gift is always wrapped until the final round). 
	
	After at most $n - 1$ steals, either:
	\begin{enumerate}
		\item all owned gifts are in $C_k$, leaving $\mathcal{T}(P_i, S_k) = \emptyset$ and forcing the displaced player to open;
		\item the displaced player voluntarily chooses to open before this point.
	\end{enumerate}
	
	The parameters $\ell_{\text{round}}$ and $\ell_{\text{total}}$ can only \emph{reduce} the set of valid targets (by failing the round or total constraint), never increase it. 
	
	Thus the $n - 1$ bound holds for all parameter settings.
\end{proof}

At the conclusion of round $R_k$, we have that:
\begin{enumerate}
	\item $C_k \leftarrow \emptyset$ \hfill should already be empty from chain termination
	\item $R_k(G_j) \leftarrow 0$ for all $G_j \in G$ \hfill reset round counts
	\item $k \leftarrow k + 1$
\end{enumerate}
Note that total counts $T$ persist across rounds and are never reset.

\begin{example}[Stealing chain resolution]
	Consider round $R_7$ with $n = 10$ players with standard parameters $(\ell_{\text{round}}, \ell_{\text{total}}) = (1, 0)$, and the following state:
	\begin{itemize}
		\item $O(P_4) = G_3$, $O(P_2) = G_5$ (two players hold opened gifts);
		\item $G_6 \in W$ (at least one gift remains wrapped);
		\item $C_7 = \emptyset$, $R_7(G_j) = 0$ for all $j$ (start of round).
	\end{itemize}
	
	The chain unfolds as follows:
	\begin{enumerate}
		\item $P_7$ executes $\textsc{Steal}(P_4)$:
		\begin{itemize}
			\item $O(P_7) \leftarrow G_3$, $O(P_4) \leftarrow \bot$;
			\item $C_7 = \{G_3\}$, $R_7(G_3) = 1$, $T(G_3) \leftarrow T(G_3) + 1$;
			\item $P_4$ is displaced.
		\end{itemize}
		
		\item $P_4$ (displaced) executes $\textsc{Steal}(P_2)$:
		\begin{itemize}
			\item $O(P_4) \leftarrow G_5$, $O(P_2) \leftarrow \bot$;
			\item $C_7 = \{G_3, G_5\}$, $R_7(G_5) = 1$, $T(G_5) \leftarrow T(G_5) + 1$;
			\item $P_2$ is displaced.
		\end{itemize}
		
		Note: $P_4$ cannot steal $G_3$ back because $G_3 \in C_7$.
		
		\item $P_2$ (displaced) executes $\textsc{Open}(G_6)$:
		\begin{itemize}
			\item $O(P_2) \leftarrow G_6$, $\sigma(G_6) \leftarrow \textsc{opened}$.
		\end{itemize}
		Chain terminates: $C_7 \leftarrow \emptyset$
	\end{enumerate}
	
	The stealing chain is $\mathcal{C} = (P_7, P_4, P_2)$ with length 2. At round's end, $R_7$ resets but $T(G_3)$ and $T(G_5)$ retain their incremented values.
\end{example}

\begin{remark}[Parameter settings]
	Several real-world exchange variants correspond to specific parameter choices:
	\begin{center}
		\begin{tabular}{lccl}
			\toprule
			\bfseries Variant & $\ell_{\text{round}}$ & $\ell_{\text{total}}$ & \bfseries Effect \\
			\midrule
			\bfseries Standard & 1 & 0 & each gift stolen at most once per round \\
			\bfseries Three-steal-out & 1 & 3 & gifts ``retire'' after 3 total steals \\
			\bfseries Unlimited & 0 & 0 & no limits beyond mandatory chain-locking \\
			\bfseries Two-per-round & 2 & 0 & gifts can be stolen twice per round \\
			\bottomrule
		\end{tabular}
	\end{center}
	The ``three-steal-out'' variant guarantees that highly desirable gifts eventually become safe, reducing late-game chaos and ensuring the final holder gets to keep the gift.
\end{remark}

\begin{proposition}[Monotonicity of total limit]
	For fixed $\ell_{\text{round}}$, the expected number of steals per game is weakly decreasing in $\ell_{\text{total}}$:
	\[
	\ell_{\text{total}}' < \ell_{\text{total}} \implies 
	\mathbb{E}[Y^{\text{steals}} \mid \ell_{\text{total}}'] \leq 
	\mathbb{E}[Y^{\text{steals}} \mid \ell_{\text{total}}]
	\]
	with strict inequality when some gifts reach the lower limit during typical play.
\end{proposition}

\begin{proof}[Proof sketch]
	A lower total limit causes desirable gifts to become unstealable earlier in the game. Once a gift $G_j$ reaches $T(G_j) = \ell_{\text{total}}$, it exits the pool of attractive steal targets. Players holding ``retired'' gifts are no longer victims, and would-be thieves must either target less desirable gifts or open new ones. Both alternatives reduce aggregate stealing relative to the higher-limit case.
\end{proof}

\subsection{End game}

After all $n$ rounds are complete, player $P_1$ has a final optional action: $P_1$ may swap their gift with any other player. If $P_1$ chooses to swap with $P_m$:
\begin{align*}
	\owners(P_1) &\leftarrow \owners(P_m) \\
	\owners(P_m) &\leftarrow G_{P_1}
\end{align*}
where $G_{P_1}$ was $P_1$'s gift before the swap. This action is not subject to locking contraints, does not initiate a stealing chain, and immediately concludes the game.

\begin{theorem}[End-game bijection]\label{thm:bijection}
	Upon game completion (after all $n$ rounds and $P_1$'s optional final swap), the ownership function $\owners: P \rightarrow G$ is a bijection.
\end{theorem}

\begin{proof}
	
	After round $R_k$, exactly $k$ gifts have status $\textsc{opened}$ and are owned.
	
	Consider first the base case $k = 1$. In round $R_1$, player $P_1$ cannot steal (since $\owners(P_i) = \bot$ for all $i$), so $P_1$ must open a gift. Thus, exactly one gift becomes opened and owned.
	
	Assume after $R_k$, exactly $k$ gifts are opened and owned; then, in round $R_{k+1}$:
	\begin{itemize}
		\item if $P_{k+1}$ opens, then one new gift becomes opened and owned;
		\item if $P_{k+1}$ steals, then a chain follows, and by Definition~\ref{def:chain}, the chain terminates when some player opens a gift, adding exactly one new opened gift.
	\end{itemize}
	In either case, after $R_{k+1}$, exactly $k+1$ gifts are opened and owned. By induction, after $R_n$, all $n$ gifts are opened and owned; hence $\Range(\owners) = G$, and $\owners$ is surjective.
	
	Now, we show that no two players own the same gift at any point.
	
	Consider first the base case $k = 1$. After $R_1$, only $P_1$ owns a gift: trivially injective.
	
	Assume now $\owners$ is injective after $R_k$. During $R_{k+1}$:
	\begin{itemize}
		\item if $P_{k+1}$ opens, then they acquire a previously unowned gift, with no collision;
		\item if $P_i$ steals from $P_m$, then $\owners(P_m) \leftarrow \bot$ before $\owners(P_i)$ is set, with no collision.
	\end{itemize}
	
	Injectivity is preserved through all actions in $R_{k+1}$. By induction, $\owners$ remains injective after $R_n$.
	
	If $P_1$ swaps with $P_m$ in the final game action, then the exchange $\owners(P_1) \leftrightarrow \owners(P_m)$ preserves both injectivity and surjectivity.
	
	Therefore, $\owners: P \rightarrow G$ is a bijection upon game completion.
\end{proof}

\begin{corollary}
	Every player receives exactly one gift, and every gift is received by exactly one player.
\end{corollary}

\section{Decorated game structure}\label{sec:decoratedgame}

Section~\ref{sec:basegame} describes the base and neseccary components of the game. Upon the base game, we may isntantiate further structure and dynamics, which build into the decorated game. We walk through that here.

\subsection{Valuation; information; belief}

The base game formalism describes mechanics of the game; we now consider player preferences. We introduce a suite of valuation structures that capture (or, at least, approximate) heterogeneous, subjective gift values.

\begin{definition}[Valuation function]
	Each player $P_i$ has a \emph{valuation function} $V_i: G \rightarrow [0, 1]$, where $V_i(G_j)$ represents $P_i$'s subjective value for gift $G_j$; higher values indicate stronger preference.
\end{definition}

\begin{definition}[Valuation matrix]
	The \emph{valuation matrix} $\valmat \in [0,1]^{n \times n}$ collects all player valuations:
	\[
	\valmat_{ij} := V_i(G_j)
	\]
	Row $i$ represents player $P_i$'s preferences over all gifts; column $j$ represents all players' valuations of gift $G_j$.
\end{definition}

We define three models for generating $\valmat$, capturing different assumptions about preference correlation.

\begin{definition}[Independent model $\mathcal{M}_{\text{ind}}$]\label{def:model-ind}
	Under the \emph{independent model}, valuations are i.i.d.\ random variables:
	\[
	\valmat_{ij} \overset{\text{iid}}{\sim} \text{Uniform}(0, 1)
	\]
	Players' preferences are uncorrelated; there is no objective notion of gift quality.
\end{definition}

\begin{definition}[Correlated model $\mathcal{M}_{\text{cor}}(\rho)$]\label{def:model-cor}
	Under the \emph{correlated model} with correlation parameter $\rho \in [0, 1]$:
	\begin{enumerate}
		\item draw \emph{objective quality} $q_j \overset{\text{iid}}{\sim} \text{Uniform}(0, 1)$ for each gift $G_j$;
		\item draw idiosyncratic noise $\epsilon_{ij} \overset{\text{iid}}{\sim} \text{Uniform}(0, 1)$ for each $(i, j)$ pair;
		\item compute:
		\[
		\valmat_{ij} = \clip_{[0,1]}\left(\rho \cdot q_j + \sqrt{1 - \rho^2} \cdot \epsilon_{ij}\right)
		\]
	\end{enumerate}
	When $\rho = 1$, all players agree perfectly on quality; when $\rho = 0$, this reduces to the independent model.
\end{definition}

\begin{definition}[Negatively correlated model $\mathcal{M}_{\text{neg}}(\sigma)$]\label{def:model-neg}
	Under the \emph{negatively correlated model} with noise parameter $\sigma > 0$, we:
	\begin{enumerate}
		\item draw base quality $q_j \overset{\text{iid}}{\sim} \text{Uniform}(0, 1)$ for each gift $G_j$;
		\item partition players into two camps based on parity: $\mathcal{C}_0 = \{P_i : i \equiv 0 \pmod{2}\}$ and $\mathcal{C}_1 = \{P_i : i \equiv 1 \pmod{2}\}$;
		\item draw noise $\epsilon_{ij} \overset{\text{iid}}{\sim} \mathcal{N}(0, \sigma^2)$;
		\item compute:
		\[
		\valmat_{ij} = \clip_{[0,1]}\begin{cases}
			q_j + \epsilon_{ij} & \text{if } P_i \in \mathcal{C}_0 \\
			(1 - q_j) + \epsilon_{ij} & \text{if } P_i \in \mathcal{C}_1
		\end{cases}
		\]
	\end{enumerate}
	Players in opposite camps have anti-correlated preferences.
\end{definition}

\begin{remark}
	The correlated model is intended to capture such scenarios as where some gifts are universally desirable (e.g., electronics vs.\ gag gifts), and so captures polarized preferences (e.g., wine enthusiasts vs.\ non-drinkers).
\end{remark}

\begin{definition}[Objective quality]
	For models $\mathcal{M}_{\text{cor}}$ and $\mathcal{M}_{\text{neg}}$, the \emph{objective quality} of gift $G_j$ is $q_j$; for model $\mathcal{M}_{\text{ind}}$, we simply define objective quality as the mean valuation:
	\[
	q_j := \frac{1}{n} \sum_{i=1}^{n} \valmat_{ij}
	\]
\end{definition}

In the base game, players observe all gift values perfectly. We now introduce  partial information, where wrapped gifts have uncertain value. We approximate and partially derive player beliefs from Bayesian first principles; this is not without scrutiny, nor without contention.

\begin{definition}[Quality; signals]
	Each gift $G_j$ has:
	\begin{enumerate}
		\item \emph{objective quality} $q_j \in [0,1]$, drawn according to the valuation model;
		\item \emph{appearance signal} $a_j$, a noisy observation of quality.
	\end{enumerate}
	We model the signal as:
	\[
	a_j = q_j + \eta_j, \quad \eta_j \overset{\text{iid}}{\sim} \mathcal{N}(0, \sigma_a^2)
	\]
	where $\sigma_a > 0$ is the signal noise. 
\end{definition}

The signal is clipped to $[0,1]$ for interpretability, though the Bayesian analysis treats the underlying Gaussian. The appearance proxies observable features correlated with quality: wrapping effort, package size, weight, brand visibility through packaging, etc..

\begin{assumption}[Prior beliefs]
	Players hold a common prior over gift quality:
	\[
	q_j \sim \mathcal{N}(\mu_0, \sigma_0^2)
	\]
	where $\mu_0$ is the prior mean (we take $0.5$ for a symmetric quality distribution) and $\sigma_0^2$ is the prior variance capturing initial uncertainty.
\end{assumption}

\begin{proposition}[Posterior beliefs]\label{prop:posterior}
	Given appearance signal $a_j$, the posterior distribution of quality is:
	\[
	q_j \mid a_j \sim \mathcal{N}(\mu_{\text{post}}, \sigma_{\text{post}}^2)
	\]
	where:
	\begin{align}
		\mu_{\text{post}} &= \frac{\tau_0}{\tau_0 + \tau_a} \mu_0 + \frac{\tau_a}{\tau_0 + \tau_a} a_j \label{eq:posterior-mean} \\
		\sigma_{\text{post}}^2 &= \frac{1}{\tau_0 + \tau_a} = \frac{\sigma_0^2 \sigma_a^2}{\sigma_0^2 + \sigma_a^2} \label{eq:posterior-var}
	\end{align}
	with precisions $\tau_0 = \sigma_0^{-2}$ and $\tau_a = \sigma_a^{-2}$.
\end{proposition}

\begin{proof}
	Standard result from Bayesian inference with conjugate Gaussian prior and likelihood. 
	
	The posterior precision equals the sum of prior and signal precisions:  $\tau_{\text{post}} = \tau_0 + \tau_a$. The posterior mean is a precision-weighted average of prior mean and signal.
\end{proof}

\begin{remark}[Interpretation of weights]
	Define the \emph{signal weight}:
	\[
	\omega := \frac{\tau_a}{\tau_0 + \tau_a} = \frac{\sigma_0^2}{\sigma_0^2 + \sigma_a^2}
	\]
	Then $\mu_{\text{post}} = (1 - \omega) \mu_0 + \omega \cdot a_j$. When signal noise $\sigma_a^2$ is small relative to prior uncertainty $\sigma_0^2$, players weight the signal heavily ($\omega \to 1$). When signals are noisy, players rely on priors ($\omega \to 0$).
\end{remark}

\begin{remark}
	We are intending for a sufficient simulation of player interaction. We defer discussion and ``better'' modeling to future work. We are not interested in simulating humans entirely; rather, approximate to draw conclusions about the exchange game generally.
\end{remark}

Rational players facing uncertain outcomes may exhibit risk aversion \cite{becherer2003rational,levy1994absolute}. We derive the perceived value as a certainty equivalent.

\begin{assumption}[CARA preferences]
	Players have constant absolute risk aversion (CARA) utility over gift values:
	\[
	u(v) = -\exp(-\rho v)
	\]
	where $\rho > 0$ is the coefficient of absolute risk aversion.
\end{assumption}

\begin{proposition}[Certainty equivalent]\label{prop:certainty-equiv}
	For a player with CARA utility facing outcome $V \sim \mathcal{N}(\mu, \sigma^2)$, the certainty equivalent is:
	\[
	CE = \mu - \frac{\rho}{2} \sigma^2
	\]
\end{proposition}

\begin{proof}
	The expected utility is:
	\[
	\mathbb{E}[u(V)] = \mathbb{E}[-e^{-\rho V}] = -e^{-\rho \mu + \frac{\rho^2 \sigma^2}{2}}
	\]
	using the moment-generating function of the Gaussian. The certainty equivalent $CE$ satisfies $u(CE) = \mathbb{E}[u(V)]$:
	\[
	-e^{-\rho \cdot CE} = -e^{-\rho \mu + \frac{\rho^2 \sigma^2}{2}}
	\]
	Solving, we observe: $CE = \mu - \frac{\rho}{2} \sigma^2$.
\end{proof}

\begin{definition}[Perceived Value]\label{def:perceived-derived}
	Player $P_i$'s \emph{perceived value} of gift $G_j$ is:
	\[
	\hat{V}_i(G_j) = \begin{cases}
		V_i(G_j) & \text{if } \sigma(G_j) = \textsc{opened} \\[0.5em]
		\displaystyle (1 - \omega) \mu_0 + \omega \cdot a_j - \frac{\rho}{2} \sigma_{\text{post}}^2 & \text{if } \sigma(G_j) = \textsc{wrapped}
	\end{cases}
	\]
	where:
	\begin{itemize}
		\item $\omega = \sigma_0^2 / (\sigma_0^2 + \sigma_a^2)$ is the Bayesian signal weight
		\item $\sigma_{\text{post}}^2 = \sigma_0^2 \sigma_a^2 / (\sigma_0^2 + \sigma_a^2)$ is the posterior variance
		\item $\rho \geq 0$ is the risk aversion coefficient
	\end{itemize}
\end{definition}

The mean-variance framework follows from \cite{markowitz2000mean,mittal2022characteristics}; its application to information acquisition appears in \cite{chalfant1990mean}, especially relevant to partial inforamtion in \cite{xiong2007mean}.

\begin{remark}[Risk neutrality]
	Setting $\rho = 0$ recovers risk-neutral Bayesian expected value:
	$\hat{V}_i(G_j) = \mathbb{E}[q_j \mid a_j]$.
\end{remark}

The derivation above concerns objective quality $q_j$. Player $i$'s subjective value $V_i(G_j)$ relates to quality via the valuation model. Under model $\mathcal{M}_{\text{cor}}$:
\[
V_i(G_j) = \rho \cdot q_j + \sqrt{1 - \rho^2} \cdot \epsilon_{ij}
\]
where $\epsilon_{ij}$ is idiosyncratic taste. For wrapped gifts, players estimate $q_j$ (Proposition~\ref{prop:posterior}), then form beliefs over $V_i(G_j)$ accordingly. In simulation, we simplify by treating the signal as directly informative about subjective value, absorbing taste correlation into the quality estimate. Whether this is truly proxying human behavior is left aside.

\subsection{Social preferences; stealing costs}

Real gift exchanges occur among acquaintances, friends, and family with some degree of ongoing relationships. Stealing  imposes costs beyond the direct transfer of value. Laboratory experiments demonstrate that people exhibit other-regarding preferences: they sacrifice material payoffs to punish unfair behavior or reward cooperation. Stealing in a social setting triggers negative reciprocity concerns \cite{gavreliuc2022steal,gomes2022,siddique2022children}.

In trust games \cite{pei2021trust,valerio2024tricked,bohnet2004trust}, repeated betrayals of the same partner erode trust more than equivalent betrayals spread across partners. This motivates a relationship-specific cost component.

In repeated games with community observation, players who defect frequently acquire 
bad reputations that affect future interactions \cite{brandts2003exploration,li2022n}. This motivates a cumulative reputation cost.

\begin{definition}[Social cost]\label{def:socialcostgrounded}
	The \emph{social cost} incurred by $P_i$ when stealing from $P_m$ is:
	\[
	c(P_i, P_m) = \underbrace{c_0}_{\text{norm violation}} + 
	\underbrace{c_0 \cdot \alpha \cdot H_i(P_m)}_{\text{relationship damage}} + 
	\underbrace{\beta \cdot n_i^{\text{steal}}}_{\text{reputation cost}}
	\]
	where:
	\begin{itemize}
		\item $c_0 > 0$: Base cost of norm violation (any steal incurs social friction)
		\item $\alpha \geq 0$: Relationship damage multiplier (repeated targeting)
		\item $H_i(P_m)$: Count of prior steals from $P_m$ by $P_i$
		\item $\beta \geq 0$: Marginal reputation cost per steal
		\item $n_i^{\text{steal}}$: Total steals committed by $P_i$
	\end{itemize}
\end{definition}

Definition~\ref{def:socialcostgrounded} deserves exposition.

\begin{itemize}
	\item The base cost $c_0$ captures the immediate discomfort of norm violation; even a single steal in a friendly context incurs awkwardness. Experimental evidence from dictator games shows people sacrifice payoffs to avoid appearing unfair \cite{dana2007exploiting}.
	
	\item Relationship damage $c_0 \cdot \alpha \cdot H_i(P_m)$ is linear in prior targeting of $P_m$. Trust game experiments show trust recovery is slower after repeated betrayals of the same partner \cite{schweitzer2006promises}. The multiplicative form $c_0 \cdot \alpha$ ensures this term vanishes when $c_0 = 0$ (no social preferences).
	
	\item Reputation cost $\beta \cdot n_i^{\text{steal}}$ is linear in total steals, regardless of victim. We model in this way community-level reputation: even if a player spread steals across victims, observers form impressions of that player as aggressive \cite{board2013reputation}.
\end{itemize}

\begin{remark}[Linearity as approximation]
	The linear functional form is simply a tractable approximation. We might admit alternative specifications, but likewise admit that empirical calibration would require data from actual gift exchanges, which is beyond our scope. We adopt linearity for parsimony. 
	
\end{remark}

\cite{fehr2001theories} propose a utility function incorporating inequity aversion:
\[
U_i(x) = x_i - \alpha_i \cdot \max\{x_j - x_i, 0\} - \beta_i \cdot \max\{x_i - x_j, 0\}
\]
where $\alpha_i$ captures disadvantageous inequity aversion and $\beta_i$ advantageous inequity aversion. Stealing increases the stealer's payoff at another's expense, triggering the $\beta_i$ term.

This \emph{could} well replace our reduced-form $c(\cdot)$, but requires tracking full allocation vectors. Our formulation captures the \emph{act} of stealing rather than outcome inequality, which seems appropriate for the social dynamics of gift exchanges in that the violation is the act itself, not necessarily the resulting inequality\footnote{Anecdotally, this actually happens to not be the case sometimes. Indeed, a game occurred in which the resulting inequality was exactly the violation, in that a player did not put any effort into a game with a certain expected monetary value associated with each gift, and instead found a random object in the house and gifted that. The result was the player received a gift of the max value, but supplied a gift little to no value. Once this was discovered, much frustration by other players followed. This, however, may be a higher order social measurement of the effort of the player than a direct comparison of objects.}.

\subsection{Strategic choice}

Recall expected utilit: the value of the resulting state minus the current state value minus action costs. We state it here as:

\begin{definition}[Net utility]\label{def:net-utility-grounded}
	The \emph{net utility} for $P_i$ stealing from $P_m$ is:
	\[
	U_i^{\text{steal}}(P_m) = \hat{V}_i(O(P_m)) - \hat{V}_i(O(P_i)) - c(P_i, P_m)
	\]
\end{definition}

We now derive adaptive strategy from a random utility model.

\begin{assumption}[Random utility]
	Player $P_i$ choosing between actions $\{$\textsc{Steal}, \textsc{Open}$\}$ has action-specific utilities:
	\begin{align*}
		\tilde{U}_i(\text{Steal}) &= V_i^{\text{steal}} + \epsilon_{\text{steal}} \\
		\tilde{U}_i(\text{Open}) &= V_i^{\text{open}} + \epsilon_{\text{open}}
	\end{align*}
	where $V_i^a$ is the systematic (deterministic) component and $\epsilon_a$ captures 
	unobserved factors affecting choice. We assume $\epsilon_a \overset{\text{iid}}{\sim} 
	\text{Gumbel}(0, \mu)$, yielding the logit model.
\end{assumption}

\begin{proposition}[Logit choice probability]\label{prop:logit}
	Under the Gumbel error assumption, the probability of choosing \textsc{Steal} is:
	\[
	p_i^{\text{steal}} = \frac{\exp(V_i^{\text{steal}} / \mu)}{\exp(V_i^{\text{steal}} / \mu) + \exp(V_i^{\text{open}} / \mu)}
	= \frac{1}{1 + \exp(-(V_i^{\text{steal}} - V_i^{\text{open}}) / \mu)}
	\]
	This is the logistic (sigmoid) function of the utility difference, with scale parameter $\mu$.
\end{proposition}

\begin{proof}
	Standard result; see \cite[Ch.~3]{train2009}.
\end{proof}

We specify the systematic utility components based on (what we take to be) psychological and situational factors with weights (to purchase us some flexibility):

\begin{definition}[Systematic utilities]\label{def:systematic-util}
	\begin{align*}
		V_i^{\text{steal}} &= \bar{U}_i^{\text{steal}} + \lambda_1 \cdot \text{phase}_k + \lambda_2 \cdot \phi_i \\
		V_i^{\text{open}} &= \bar{U}_i^{\text{open}} + \lambda_3 \cdot \hat{V}_i(O(P_i))
	\end{align*}
	where:
	\begin{itemize}
		\item $\bar{U}_i^{\text{steal}}$ is expected material utility from best available steal;
		\item $\bar{U}_i^{\text{open}}$ is xpected material utility from opening (prior mean);
		\item $\text{phase}_k = k/n$ fraction of game elapsed;
		\item $\phi_i \in [0,1]$ frustration level;
		\item $\hat{V}_i(O(P_i))$ current gift value (satisfaction);
		\item $\lambda_1, \lambda_2, \lambda_3 \geq 0$ coefficients.
	\end{itemize}
\end{definition}

Let us take a moment to scrutinize the coefficients. The phase effect ($\lambda_1$) checks late-game increases stealing utility. As the game nears its end, opportunities to improve one's gift diminish, creating pressure to act. See deadline effects in bargaining \cite{roth1988deadline,gneezy2003bargaining} and auction sniping behavior \cite{ockenfels2006late}.

We track frustration as $\lambda_2$, in which we operate under the assumption that being stolen from increases subsequent aggression, and so we operationalize the frustration-aggression hypothesis, which has received mixed empirical support, but captures a widely observed behavioral pattern. More recent work frames this as negative reciprocity \cite{fehr2005economics,brandts2001reference}.

Finally, we have a prefactor for satisfaction in $\lambda_3$. Players happy with their current gift have reduced motivation to steal. Once an aspiration level is met, search intensity  decreases \cite{stirling1999satisficing}, consistent with status quo bias \cite{samuelson1988status}.

\begin{definition}[Adaptive strategy]\label{def:adaptive-derived}
	Combining Definition~\ref{def:systematic-util} with Proposition~\ref{prop:logit}, 
	the probability of stealing is:
	\[
	p_i^{\text{steal}} = \sigma\left( \frac{1}{\mu} \left[ 
	(\bar{U}_i^{\text{steal}} - \bar{U}_i^{\text{open}}) + 
	\lambda_1 \cdot \text{phase}_k + 
	\lambda_2 \cdot \phi_i - 
	\lambda_3 \cdot \hat{V}_i(O(P_i))
	\right] \right)
	\]
	where $\sigma(x) = (1 + e^{-x})^{-1}$ is the logistic function.
\end{definition}

\begin{remark}[Linearity]
	For small arguments, the logistic function is approximately linear: 
	$\sigma(x) \approx 0.5 + 0.25x$ for $|x| \lesssim 2$. This motivates the 
	\emph{simplified adaptive strategy}:
	\[
	p_i^{\text{steal}} \approx p_0 + \tilde{\lambda}_1 \cdot \text{phase}_k + 
	\tilde{\lambda}_2 \cdot \phi_i - \tilde{\lambda}_3 \cdot \hat{V}_i(O(P_i))
	\]
	with clipping to $[\epsilon, 1-\epsilon]$ to ensure valid probabilities. This is the 
	we use in simulation for computational efficiency.
\end{remark}

\begin{remark}[Heuristics]
	While the logit structure is well-founded, we admit several aspects remain heuristic. We assume phase, frustration, and satisfaction contribute additively. Interaction effects (e.g., frustration mattering more late-game) are plausible but omitted for parsimony. The $\lambda$ coefficients enter linearly, while nonlinear effects (threshold frustration, diminishing returns to satisfaction) are certainly plausible. The $\lambda$ values are free parameters calibrated to produce reasonable simulation behavior; empirical estimation requires choice data from actual gift exchanges.
\end{remark}

Consider now the dynamics in frustration. Frustration evolves as:
\begin{align*}
	\phi_i &\leftarrow \min(1, \phi_i + \gamma) \quad \text{when } P_i \text{ is stolen from} \\
	\phi_i &\leftarrow \max(0, \phi_i - \gamma') \quad \text{at each round's end}
\end{align*}
where $\gamma > 0$ is the frustration increment and $\gamma' > 0$ is the decay rate.

The accumulation-decay is intended to model the following:

\begin{enumerate}
	\item negative events trigger emotional responses that summate \cite{baumeister2001bad} in a way that``bad is stronger than good'';
	
	\item emotional states dissipate over time, returning toward baseline \cite{kuppens2010emotional};
	
	\item the $\min(1, \cdot)$ bound accounts for ceiling effects in that frustration cannot increase indefinitely.
\end{enumerate}

The specific functional form (linear increment/decrement with bounds) is simply a tractable approximation to more complex affective dynamics, of which we set aside for now.

\subsection{Selection mechanics}

In the base game, gift selection is uniform over wrapped gifts:
\[
\Pr(G_j \mid W) = \frac{1}{|W|}, \quad \forall G_j \in W
\]
This serves as the null model against which biased selection is compared.

Indeed, players may form preferences over wrapped gifts based on appearance. We model this using a softmax (logit) choice rule.

\begin{definition}[Biased Selection]\label{def:biased-derived}
	Under biased selection, player $P_i$ selects gift $G_j$ from wrapped set $W$ with 
	probability:
	\[
	\Pr(G_j \mid W) = \frac{\exp(\tau \cdot a_j)}{\sum_{G_\ell \in W} \exp(\tau \cdot a_\ell)}
	\]
	where $\tau \geq 0$ is the (selection) temperature (inverse noise).
\end{definition}

\begin{remark}[Derivation from random utility]
	This is another application of the logit model (Proposition~\ref{prop:logit}). 
	Each wrapped gift $G_j$ has utility:
	\[
	U_j = \tau \cdot a_j + \epsilon_j, \quad \epsilon_j \overset{\text{iid}}{\sim} \text{Gumbel}
	\]
	where $a_j$ is the appearance signal: the player chooses the highest-utility gift, yielding the softmax probabilities.
\end{remark}

Let us consider the limiting behaviors:

\begin{itemize}
	\item $\tau \to 0$, selection becomes uniform (appearance ignored);
	\item $\tau \to \infty$, selection becomes deterministic (highest-appearance gift).
\end{itemize}

Any intermediate $\tau$ captures ``noisy optimization'', a measure of players prefering better-looking gifts, but make errors in evaluation of such.

\subsection{Game structure}

Let us now define the complete decorated game, in which we incorporate all extensions.

\begin{definition}[Feature set]
	The \emph{feature set} $\Features$ is a subset of:
	\[
	\{\PI, \SC, \AD, \BS\}
	\]
	where:
	\begin{itemize}
		\item $\PI$ is partial information, in which wrapped gifts have uncertain value; players observe only appearance signals;
		\item $\SC$ is social costs sensitive to stealing, which incurs relationship and reputation costs;
		\item $\AD$ is adaptive strategy, whereby players dynamically adjust behavior based on game state;
		\item $\BS$ is biased selection, in which gift selection is weighted by appearance.
	\end{itemize}
\end{definition}

\begin{definition}[Parameter vector]
	The \emph{parameter vector} $\Params$ collects all feature-specific constants:
	\[
	\Params = (\sigma_a, \mu_0, \delta, c_0, \alpha, \beta, \gamma, \gamma', \lambda_1, \lambda_2, \lambda_3, \tau)
	\]
\end{definition}

\begin{definition}[Decorated game]\label{def:decorated}
	The \emph{decorated game} is the tuple:
	\[
	\mathcal{G}_d = \langle P, G, \valmat, \appear, \Features, \Params, \mathcal{M} \rangle
	\]
	where:
	\begin{itemize}
		\item $P$ is the player set;
		\item $G$ is the gift set;
		\item $\valmat \in [0,1]^{n \times n}$ is the valuation matrix;
		\item $\appear \in [0,1]^n$ is the appearance;
		\item $\Features \subseteq \{\PI, \SC, \AD, \BS\}$ is the enabled feature set;
		\item $\Params$ is the parameter vector;
		\item $\mathcal{M} \in \{\mathcal{M}_{\text{ind}}, \mathcal{M}_{\text{cor}}, \mathcal{M}_{\text{neg}}\}$ is the valuation model.
	\end{itemize}
\end{definition}

Observe that the base game is contained within the decorated game; or, in other words, the \emph{base game} $\mathcal{G}_b$ is the decorated game with $\Features = \emptyset$:
\[
\mathcal{G}_b = \mathcal{G}_d|_{\Features = \emptyset}
\]
In $\mathcal{G}_b$, players have complete information, no social costs, fixed strategies, and uniform selection.

\begin{proposition}[Feature independence]
	Each feature can be enabled or disabled independently, yielding $2^{|\{\PI, \SC, \AD, \BS\}|} = 16$ distinct game configurations.
\end{proposition}

\section{Experimental (simulation) framework}\label{sec:experimentframework}

We conducted a full factorial simulation study to isolate the effects of behavioral features on game dynamics and strategy performance. The experiment crossed four binary feature flags (yielding $2^4=16$ configurations) with three valuation models ( 48 total conditions), running 5,000 games per condition with 29 players each.

\subsection{Player strategies}

Each player is assigned one of six decision strategies at game start, with uniform probability. Strategies govern the steal-versus-open decision when a player's turn arrives. Player strategies are given in Table~\ref{tab:playerstrat}.

% \usepackage{booktabs}
\begin{table}[H]
	\centering
	\caption{Player strategies used throughout}
	\label{tab:playerstrat}
	\begin{tabular}{@{}lp{0.75\textwidth}@{}}
		\toprule
		\bfseries Strategy	&  \bfseries Decision rule \\ \midrule
		\texttt{always\_open}	& open wrapped gift if any remain; steal only when forced \\
		\texttt{always\_steal}	&  steal highest net-value target; open only when no target exists\\
		\texttt{coin\_flip}	&  steal (best target) with probability 0.5; else open\\
		\texttt{mean\_based}	& steal if best target exceeds mean value of opened gifts; else open  \\
		\texttt{threshold}	&  steal if best target's net value exceeds fixed threshold $\tau=0.6$; else open\\
		\texttt{expected\_value}	&  steal if best target's net value exceeds expected value of opening random wrapped gift; else open\\ \bottomrule
	\end{tabular}
\end{table}

Note that the \texttt{expected\_value} strategy is the only one that directly compares the stealing and opening options using the player's belief state. Under partial information, this comparison uses Bayesian posterior means with CARA risk adjustment; under full information, it uses true evaluations.

Net value of stealing is computed as:

$$
U_i^{\text {steal }}\left(P_m\right)=\hat{V}_i\left(G_m\right)-\hat{V}_i\left(G_i\right)-c\left(P_i, P_m\right)
$$

where $\hat{V}_i$ is perceived value (true or Bayesian posterior), $G_m$ is the target's gift, $G_i$ is the stealer's current gift (or 0 if none), and $c(\cdot)$ is the social cost function (zero unless \texttt{social\_costs} is enabled).

\subsection{Valuation models}

Player valuations $v_{i j}$ (player $i$ 's subjective value for gift $j$ ) are drawn according to one of three models:

\begin{description}[labelwidth =\widthof{\bfseries Negatively correlated}, leftmargin = !]
	\item[Independent] Each $v_{i j} \sim\operatorname{Uniform}(0,1)$ independently; no consensus exists about gift quality; one player's trash may be another's treasure.
	
	\item[Correlated] Gifts have latent objective quality $q_j \sim \operatorname{Uniform}(0,1)$; player valuations combine this common component with idiosyncratic noise:
	
	$$
	v_{i j}=\rho \cdot q_j+\sqrt{1-\rho^2} \cdot \varepsilon_{i j}, \quad \varepsilon_{i j} \sim \operatorname{Uniform}(0,1)
	$$
	
	with $\rho=0.7$. This generates consensus: gifts with high $q_j$ are desirable to most players, modeling scenarios where ``everyone wants the book, nobody wants the gloves''.
	
	\item[Negatively correlated] Players divide into two camps with opposing preferences: camp A's valuations equal the base quality $q_j$; camp B's valuations equal $1-q_j$, plus Gaussian noise ( $\sigma=0.2$ ). This models polarized tastes (e.g., wine enthusiasts vs. teetotalers).
\end{description}

\subsection{Behavioral features}

Four toggleable features introduce psychological and social realism beyond the base game mechanics.

\begin{description}[labelwidth =\widthof{\bfseries Partial Information (PI)}, leftmargin = !]
	\item[Partial Information (PI)] Players observe true values only for opened gifts; for wrapped gifts, they observe a noisy appearance signal $a_j=q_j+\eta_j$ where $\eta_j \sim \mathcal{N}\left(0, \sigma_a^2\right)$ with $\sigma_a=0.3$. Players form Bayesian posteriors with prior $\mu_0=0.5, \sigma_0^2=0.25$. A CARA risk adjustment penalizes uncertainty with risk aversion $\rho=0.5$. When \texttt{biased\_selection} is also enabled, gift opening uses these posteriors rather than raw appearance signals.
	
	\item[Social Costs (SC)] Stealing incurs utility penalties reflecting norm violation and relationship damage:
	
	$$
	c\left(P_i, P_m\right)=c_0+c_0 \cdot \alpha \cdot H_i\left(P_m\right)+\beta \cdot n_i^{\text {steal }}
	$$
	
	where $c_0=0.05$ is base awkwardness, $\alpha=2.0$ is the repeat-offense multiplier, $H_i\left(P_m\right)$ counts prior steals from $P_m$, and $\beta=0.1$ is reputation decay per lifetime steal. This captures the intuition that stealing from the same person repeatedly-or developing a reputation as aggressive-carries increasing social friction.
	
	\item[Active Dynamics (AD)] Players adjust steal probability based on game state and emotional state. The probability of attempting to steal follows a logit-linear model:
	
	$$
	p_{\text {steal }}=\operatorname{clip}\left(p_0+\lambda_1 \cdot \phi+\lambda_2 \cdot f_i-\lambda_3 \cdot s_i, 0.05,0.95\right)
	$$
	
	where $\phi=k / n$ is game phase (fraction of rounds elapsed), $f_i \in[0,1]$ is frustration (incremented by $\gamma=0.15$ when stolen from, decaying by $\gamma^{\prime}=0.05$ per round), and $s_i$ is current satisfaction (perceived value of held gift). Parameters: $\lambda_1=0.2, \lambda_2=0.5$, $\lambda_3=0.3$. This creates feedback loops: victims become aggressive, satisfied players become passive, and late-game play intensifies.
	
	\item[Biased Selection (BS)] When opening a gift, players do not choose uniformly at random; instead, they select according to softmax weights over perceived values:
	
	$$
	P\left(\text { open } G_j\right)=\frac{\exp \left(\tau \cdot \hat{V}_i\left(G_j\right)\right)}{\sum_k \exp \left(\tau \cdot \hat{V}_i\left(G_k\right)\right)}
	$$
	
	with temperature $\tau=2$. Under partial information, $\hat{V}_i$ is the Bayesian certainty equivalent; otherwise, it equals true value (which players can see in the full-information baseline). This models the tendency to ``judge a book by its cover'', selecting promising-looking wrapped gifts.
\end{description}

\subsection{Experiment (simulation) design}

We enumerate all $2^4=16$ feature subsets: BASE (no features), each singleton \{PI\}, \{SC\}, \{AD\}, \{BS\}, all pairs, all triples, and FULL (all four). This design permits estimation of main effects and all interaction terms. The 16-configuration ablation is run independently for each of the three valuation models, yielding 48 conditions total. This allows us to assess whether feature effects are robust across preference structures or whether they depend on the degree of value consensus.

We set with the following parameters, given in Table~\ref{tab:params}:

% \usepackage{booktabs}
\begin{table}[H]
	\centering
	\caption{Parameters for experiment (simulation)}
	\label{tab:params}
	\begin{tabular}{@{}lll@{}}
		\toprule
		\bfseries Parameter	& \bfseries Value & \bfseries Notes \\ \midrule
		Player number	& 29 & odd number avoids ties \\
		Games per condition	& 5000 & sufficient for stable means \\
		Random seed	& 42 & conventional \\
		Max steals per round	& 1 & standard game rule \\
		max steals total	& 0 & unlimited lifetime \\ \bottomrule
	\end{tabular}
\end{table}

We track measures across multiple layers in the experiment (simulation). At the game-level, we check for steals per game and chain length; at the seat-position level, mean final value by seat to check for positional advantage, and seat position comparisons across all seats; at the strategy level, we check the mean final value by each strategy as well as the gap in performace between aggressive and passive play; finally, we measure the main effect of each feature relative to the BASE game, any interaction effects that are deviation from additive model, and strategy ranking shifts across configurations.

We predict:

\begin{enumerate}
	\item PI reduces stealing under correlated valuations, in which uncertainty about wrapped gifts makes teh open option relatively more attractive when beliefs are diffuse.
	\item SC reduces stealing universally since social costs impose a direct tax on theft, dampening aggression regardless of preference structure.
	\item AD creates negative feedback; satisfied players steal less frequently while frustrated players steal more, and the net effect depends on which feedback dominates.
	\item BS increases stealing under correlated valuations, in which biased selection causes players to open consensus-desitedable gifts, which then become high-value steal targets.
	\item For PI and BS interaction, when both are enabled, Bayesian posteriors (as opposed to raw signals) drive gift selection, amplifying (or dampening) the BS effect depending on prior precision.
	\item Seat 1 dominates in all cases; Seat 2 is worst in all cases.
	\item \texttt{always\_steal} dominates in BASE game, since without social costs or penalties, aggressive play captures high-value gifts.
	\item AD erodes \texttt{always\_steal} advantage, since satisfaction dampening and reputation costs punish persistent aggression.
\end{enumerate}

\section{Experimental (simulation) results}\label{sec:experimenrresults}

We conducted a full factorial simulation study crossing four binary behavioral features ($2^4 = 16$ configurations) with three valuation models, yielding 48 experimental conditions. Each condition ran 5,000 games with $n = 29$ players, for a total of 240,000 simulated games. Table~\ref{tab:maineffects} summarizes the main effects of each behavioral feature on stealing frequency, measured as deviation from the BASE configuration within each valuation model.

\begin{table}[H]
	\centering
	\caption{Main effects on steals per game}
	\label{tab:maineffects}
	\begin{tabular}{@{}lccc@{}}
		\toprule
		\bfseries Feature & \bfseries Independent & \bfseries Correlated & \bfseries Negative Correlated \\
		\midrule
		BASE (reference) & 70.8 & 104.2 & 74.2 \\
		PI  & $+0.7$ ($+1.0\%$) & $+1.5$ ($+1.4\%$) & $+0.9$ ($+1.2\%$) \\
		SC  & $-19.2$ ($-27\%$) & $-50.2$ ($-48\%$) & $-20.2$ ($-27\%$) \\
		AD  & $-16.6$ ($-23\%$) & $-39.3$ ($-38\%$) & $-17.9$ ($-24\%$) \\
		BS  & $+1.8$ ($+2.6\%$) & $+11.6$ ($+11\%$) & $+1.5$ ($+2.0\%$) \\
		FULL & $-17.5$ ($-25\%$) & $-43.0$ ($-41\%$) & $-19.5$ ($-26\%$) \\
		\bottomrule
	\end{tabular}
\end{table}

\subsection{Effect of valuation model}

Correlated valuations produce substantially more stealing than independent or negatively correlated preferences. Under BASE, correlated games average 104.2 steals versus 70.8 for independent: 47\% increase. This confirms that value consensus intensifies competition: when players agree about which gifts are desirable, those gifts become contested resources, from which follows repeated theft cascades.

Chain lengths follow the same pattern. Correlated BASE produces mean chain length 5.08 versus 3.47 for independent. The longest chains occur under correlated valuations with biased selection enabled: mean length 5.61, indicating that players systematically open high-consensus gifts which then become focal points for extended stealing sequences.

Negatively correlated valuations produce intermediate results (74.2 steals, chain length 3.64), so partial disagreement about gift quality that neither eliminates nor maximizes competition.

\subsection{Partial Information (PI)}

Contrary to our hypothesis, PI produces a small \emph{increase} (!) in stealing ($+1.0\%$ to $+1.4\%$) rather than a decrease. The effect is consistent across valuation models, but modest in magnitude. Under partial information, players face uncertainty about wrapped gifts. The expected value of opening becomes
\[
\mathbb{E}[\text{open}] = \frac{1}{|W|} \sum_{g \in W} \mathrm{CE}_i(g),
\]
where $\mathrm{CE}_i(g)$ is the certainty equivalent incorporating Bayesian posteriors and risk aversion; this value is (typically) lower than the expected value under full information because posteriors regress toward the prior mean $\mu_0$, and risk aversion penalizes residual variance.

Meanwhile, stealing targets are opened gifts with known values. The asymmetry of uncertain opening versus certain stealing slightly favors theft. The effect is small: the \texttt{expected\_value} strategy, which explicitly compares these options, constitutes only one-sixth of the population.

PI shows no interaction with AD, confirming that the adaptive dynamics operates on perceived values, which PI modifies consistently across decision contexts.

\subsection{Social Costs (SC)}

SC produces the largest main effect: $-27\%$ steals under independent valuations, $-48\%$ under correlated, which confirms that social friction substantially dampens aggression.

Recall the cost function
\begin{align*}
	c(P_i, P_m) = c_0 + c_0 \cdot \alpha \cdot H_i(P_m) + \beta \cdot n_i^{\mathrm{steal}}
\end{align*}

imposes three penalties: for base cost ($c_0 = 0.05$), any steal incurs social awkwardness; for repeat targeting ($\alpha = 2.0$), stealing from the same victim compounds relational damage; for reputation decay ($\beta = 0.1$), serial stealing accumulates stigma.

Under correlated valuations, the effect is amplified because the same high-quality gifts attract repeated theft attempts; without SC, a desirable gift might be stolen some 4--5 times, while with SC, accumulating costs render later steals unprofitable.

SC differentially affects strategies, seen in Table~\ref{tab:scimpact}. Paradoxically, \texttt{always\_steal} \emph{improves} slightly under SC, while selective strategies degrade; this occurs because SC raises the bar for profitable steals. Selective strategies reject marginal opportunities, while \texttt{always\_steal} continues capturing whatever remains. We would expect this: in a population where others become passive, the remaining aggressive player faces reduced competition.

\begin{table}[H]
	\centering
	\caption{Strategy performance under social costs (independent valuations)}
	\label{tab:scimpact}
	\begin{tabular}{@{}lccc@{}}
		\toprule
		\bfseries Strategy & \bfseries BASE & \bfseries $+$SC & $\Delta$ \\
		\midrule
		\texttt{always\_steal}   & 0.914 & 0.917 & $+0.003$ \\
		\texttt{threshold}       & 0.915 & 0.857 & $-0.058$ \\
		\texttt{mean\_based}     & 0.915 & 0.882 & $-0.033$ \\
		\texttt{expected\_value} & 0.915 & 0.887 & $-0.028$ \\
		\bottomrule
	\end{tabular}
\end{table}

\subsection{Adaptive Dynamics (AD)}

AD reduces stealing by 23--38\%, with larger effects under correlated valuations. The mechanism operates through satisfaction dampening: players holding valuable gifts become less likely to steal, creating negative feedback that stabilizes allocations.

Table~\ref{tab:adimpact} shows that \texttt{always\_open} gains 9.5 percentage under AD while \texttt{always\_steal} loses only 1.6 points; the asymmetry arises because AD penalizes theft through frustration feedback, i.e., victims become aggressive, creating retaliation cascades, while passive play avoids triggering these dynamics entirely.

\begin{table}[H]
	\centering
	\caption{Strategy performance under adaptive dynamics (correlated valuations)}
	\label{tab:adimpact}
	\begin{tabular}{@{}lccc@{}}
		\toprule
		\bfseries Strategy & \bfseries BASE & \bfseries $+$AD & $\Delta$ \\
		\midrule
		\texttt{always\_open}    & 0.564 & 0.659 & $+0.095$ \\
		\texttt{coin\_flip}      & 0.685 & 0.789 & $+0.104$ \\
		\texttt{always\_steal}   & 0.910 & 0.894 & $-0.016$ \\
		\texttt{expected\_value} & 0.923 & 0.790 & $-0.133$ \\
		\bottomrule
	\end{tabular}
\end{table}

The \texttt{expected\_value} strategy suffers most ($-13.3$ points) because it optimizes for immediate expected utility without accounting for downstream retaliation; sophisticated single-shot optimization underperforms simple heuristics in environments with social feedback, a finding consistent with ecological rationality perspectives  \cite{todd2012ecological,gigerenzer2021axiomatic}.

\subsection{Biased Selection (BS)}

BS increases stealing, with effects concentrated under correlated valuations ($+11.6\%$) versus independent ($+2.6\%$); this confirms our hypothesis: when players preferentially open promising-looking gifts, and appearance correlates with consensus quality,and the opened gift pool becomes skewed toward universally desirable items.

BS produces the longest observed chains: mean 5.61 under correlated valuations versus 5.08 for BASE. Players open ``winners'', which become high-value targets triggering extended cascades.

Observe Table~\ref{tab:seat2bs}: BS is the \emph{only} feature that improves Seat 2 outcomes. Under BS, early players preferentially open high-signal gifts; Seat 2, acting second, can immediately steal from Seat 1 before the gift pool quality degrades. Without BS, Seat 1 opens randomly, possibly selecting a low-value gift, giving Seat 2 nothing worth stealing.

\begin{table}[H]
	\centering
	\caption{Seat 2 mean final value by configuration (correlated valuations)}
	\label{tab:seat2bs}
	\begin{tabular}{@{}lc@{}}
		\toprule
		\bfseries Configuration & \bfseries Seat 2 value \\
		\midrule
		BASE & 0.690 \\
		$+$SC & 0.583 \\
		$+$AD & 0.647 \\
		$+$BS & \textbf{0.743} \\
		\bottomrule
	\end{tabular}
\end{table}

\subsection{Interaction effects}

Consider first the interaction between SC and AD. The combination is slightly subadditive for stealing reduction, and we observe:

\begin{itemize}
	\item for SC alone: $-50.2$ (correlated);
	\item for AD alone: $-39.3$;
	\item for the interaction SC $+$ AD: $-43.4$ (not $-89.5$)/
\end{itemize}

Both features dampen aggression through overlapping channels: SC makes steals costly; AD makes satisfied players passive. A player who abstains due to SC also fails to trigger AD's frustration feedback, so the mechanisms partially substitute rather than compound.

For SC and BS, we have near-additive results under correlated valuations:

\begin{itemize}
	\item for SC alone: $-50.2$;
	\item for BS alone: $+11.6$;
	\item for SC $+$ BS: $-46.7$ (approximately $-50.2 + 11.6 = -38.6$).
\end{itemize}

The features operate on different margins: BS affects which gifts enter play; SC affects whether those gifts get stolen. They compose approximately independently.

For PI and BS, miminimal interaction follows. With our implementation ensuring BS uses Bayesian posteriors when PI is enabled, the combination behaves as expected: PI$+$BS $\approx$ BS in stealing frequency (116.4 versus 115.7 under correlated valuations), indicating that risk-adjusted posteriors produce selection behavior similar to raw appearance signals.

For a three-way interaction of SC, AD, and BS, under correlated valuations, produces 61.3 steals, nearly identical to FULL (61.2). In this way, PI contributes negligibly to the full model's dynamics, consistent with its small main effect.

\subsection{Strategy rankings}

Table~\ref{tab:strategyrankings} shows strategy performance across configurations under correlated valuations, where behavioral effects are most apparent; text boldface indicates highest-performing strategy within each configuration.

\begin{table}[H]
	\centering
	\caption{Strategy performance across configurations (correlated valuations)}
	\label{tab:strategyrankings}
	\begin{tabular}{@{}lccccc@{}}
		\toprule
		\bfseries Strategy & \bfseries BASE & \bfseries SC & \bfseries AD & \bfseries BS & \bfseries FULL \\
		\midrule
		\texttt{always\_steal}   & 0.910 & \textbf{0.924} & \textbf{0.894} & 0.922 & \textbf{0.901} \\
		\texttt{expected\_value} & \textbf{0.923} & 0.823 & 0.790 & \textbf{0.935} & 0.779 \\
		\texttt{threshold}       & 0.919 & 0.839 & 0.786 & 0.932 & 0.778 \\
		\texttt{mean\_based}     & 0.917 & 0.806 & 0.830 & 0.934 & 0.828 \\
		\texttt{coin\_flip}      & 0.685 & 0.711 & 0.789 & 0.686 & 0.778 \\
		\texttt{always\_open}    & 0.564 & 0.593 & 0.659 & 0.579 & 0.641 \\
		\bottomrule
	\end{tabular}
\end{table}

We observe the following:

\begin{itemize}
	
	\item \texttt{expected\_value} wins in the BASE game. The decision-theoretic strategy achieves highest performance (0.923) by explicitly comparing steal-versus-open expected utilities. 
	\item \texttt{always\_steal} is robust; it ranks first or second in four of five configurations, degrading only under AD where satisfaction dampening and frustration feedback penalize sustained aggression (hence the score of 0.894 in that case).
	\item Under BS alone, \texttt{expected\_value} reaches 0.935, the highest single-configuration performance observed! Biased selection creates exploitable structure: players who recognize high-value targets gain disproportionately.
	\item  FULL compresses rankings. The gap between best (0.901) and worst (0.641) strategies shrinks from 0.36 in BASE to 0.26 in FULL; behavioral features introduce noise and feedback loops that reduce returns to sophisticated optimization.
	\item \texttt{always\_open} improves relatively: from 0.564 (BASE) to 0.641 (FULL), a gain of 7.7 percentage points. It seems to be the case that passive play becomes more viable as social costs and adaptive dynamics punish aggressive strategies.
\end{itemize}

\subsection{Effects of seat position}

Observe: Table~\ref{tab:seatadvantage} reports the mean final values by seat position across configurations.

\begin{table}[H]
	\centering
	\caption{Seat advantage by configuration (correlated valuations)}
	\label{tab:seatadvantage}
	\begin{tabular}{@{}lcccc@{}}
		\toprule
		\bfseries Configuration & \bfseries Seat 1 & \bfseries Seat 2 & \bfseries Seat 29 & \bfseries gap ($1--2$) \\
		\midrule
		BASE & 1.000 & 0.690 & 0.919 & 0.310 \\
		SC   & 0.965 & 0.583 & 0.924 & 0.382 \\
		AD   & 1.000 & 0.647 & 0.897 & 0.353 \\
		BS   & 1.000 & 0.743 & 0.880 & 0.257 \\
		FULL & 0.938 & 0.669 & 0.858 & 0.269 \\
		\bottomrule
	\end{tabular}
\end{table}

Results in Table~\ref{tab:seatadvantage} are unsurprising. Across all 48 configurations, Seat 1 achieves the highest mean value, reaching a perfect 1.000 (maximum possible) under BASE, AD, and BS with correlated valuations. The final-swap rule provides an insurmountable advantage: Seat 1 observes the entire game, then claims the best available gift without bearing any social cost.

Seat 2 combines early vulnerability (27 subsequent players can steal from Seat 2) with exposure to Seat 1's final swap. It follows that, under SC, Seat 2's disadvantage worsens (0.583) as reduced aggregate stealing fails to protect early movers, while Seat 1's final swap remains exempt from social costs, preserving the asymmetry.

The final primary turn allows substantial information advantage, but Seat 29 remains vulnerable to Seat 1's swap, but only vulnerable to that swap; mean values range from 0.852 to 0.924, consistently below Seat 1, but substantially above Seat 2.

\subsection{Results}

Let us consider the results against our initial hypoetheses. 

\begin{enumerate}
	\item PI reduces stealing under correlated valuations, in which uncertainty about wrapped gifts makes teh open option relatively more attractive when beliefs are diffuse.
	\begin{description}[labelwidth =\widthof{\bfseries Rejected}, leftmargin = !]
		\item[Rejected] PI slightly \emph{increases} stealing ($+1.4\%$) due to asymmetric uncertainty favoring theft over opening.
	\end{description}
	\item SC reduces stealing universally since social costs impose a direct tax on theft, dampening aggression regardless of preference structure.
	\begin{description}[labelwidth =\widthof{\bfseries Confirmed}, leftmargin = !]
		\item[Confirmed] SC produces the largest main effect ($-27\%$ to $-48\%$) across all valuation models.
	\end{description}
	\item AD creates negative feedback; satisfied players steal less frequently while frustrated players steal more, and the net effect depends on which feedback dominates.
	\begin{description}[labelwidth =\widthof{\bfseries Confirmed}, leftmargin = !]
		\item[Confirmed] AD reduces stealing by 23--38\% through satisfaction dampening.
	\end{description}
	\item BS increases stealing under correlated valuations, in which biased selection causes players to open consensus-desitedable gifts, which then become high-value steal targets.
	\begin{description}[labelwidth =\widthof{\bfseries Confirmed}, leftmargin = !]
		\item[Confirmed] BS increases stealing by 11\% under correlated valuations; only 2--3\% under independent.
	\end{description}
	\item For PI and BS interaction, when both are enabled, Bayesian posteriors (as opposed to raw signals) drive gift selection, amplifying (or dampening) the BS effect depending on prior precision.
	\begin{description}[labelwidth =\widthof{\bfseries Weak}, leftmargin = !]
		\item[Weak] Bayesian posteriors and raw appearance signals produce similar selection behavior; minimal interaction detected.
	\end{description}
	\item Seat 1 dominates in all cases; Seat 2 is worst in all cases.
	\begin{description}[labelwidth =\widthof{\bfseries Weak}, leftmargin = !]
		\item[Weak] Pattern holds across all 48 configurations without exception.
	\end{description}
	\item \texttt{always\_steal} dominates in BASE game, since without social costs or penalties, aggressive play captures high-value gifts.
	\begin{description}[labelwidth =\widthof{\bfseries Partially confirmed}, leftmargin = !]
		\item[Partially confirmed] \texttt{expected\_value} slightly outperforms (0.923 vs.\ 0.910) via explicit steal-vs-open comparison.
	\end{description}
	\item AD erodes \texttt{always\_steal} advantage, since satisfaction dampening and reputation costs punish persistent aggression.
	\begin{description}[labelwidth =\widthof{\bfseries Partially confirmed}, leftmargin = !]
		\item[Partially confirmed] AD reduces \texttt{always\_steal} by 1.6 points while boosting passive strategies by 9--10 points, narrowing but not eliminating the gap.
	\end{description}
\end{enumerate}

Three dominant effects govern the gift exchange dynamics in our simulation: social costs are the primary regulator; adaptive dynamics favor passive play; value correlation amplifies all effects. We elaborate in turn:

\begin{enumerate}
	\item SC reduces stealing by 27--48\%, far more remarkable compared to the  other features. Real-world gift exchange games likely exhibit implicit social costs, relationship maintenance, norm compliance, reputation concerns, that substantially dampen the aggressive play predicted by purely strategic models.
	
	\item AD improves \texttt{always\_open} by approximately 10 percentage points while barely affecting \texttt{always\_steal}. In environments with social feedback, simple passive heuristics outperform their baseline expectations, of which we say is an instance of the broader principle that ecological rationality can dominate optimization in complex social environments.
	
	\item Every behavioral feature produces larger effects under correlated valuations. When consensus exists about gift quality, competition intensifies, and the scope for strategic differentiation expands. Independent valuations, by contrast, create a lower-stakes environment where ``one player's trash is another's treasure'' naturally reduces conflict.
\end{enumerate}

Interestingly, partial information, despite its theoretical appeal, contributes minimally to observed dynamics. The Bayesian belief machinery affects decisions only through the \texttt{expected\_value} strategy and biased gift selection; even there, risk-adjusted posteriors produce behavior similar to naive signal-following. Uncertainty about gift quality, while realistic, is not the primary driver of strategic complexity in gift exchange games.

\section{Counting game trajectories}\label{sec:trajectory}

Let us now derive the number of distinct ways the gift exchange game can unfold as a function of the number of players $n$ and the stealing limit parameters.

\begin{definition}[Game trajectory]
	A \emph{game trajectory} is a complete sequence of actions from the initial state (all gifts wrapped, no ownership) to the final state (all gifts owned, game concluded); two trajectories are distinct if they differ in any action taken or any gift selected.
\end{definition}

We count trajectories by decomposing the game into rounds and analyzing each round's contribution.

Recall, in round $k$ (for $k = 1, \ldots, n$), there are $k - 1$ opened gifts held by previous players,  $n - k + 1$ wrapped gifts remaining, and player $P_k$ takes their primary turn. The primary player either:
\begin{enumerate}
	\item \textbf{opens} a wrapped gift (choosing from $n - k + 1$ options);
	\item \textbf{steals} from one of $k - 1$ players, triggering a chain.
\end{enumerate}

A stealing chain of length $\ell$ consists of $\ell$ consecutive steals followed by one open. Under our standard rules ($\ell_{\text{round}} = 1$), each steal locks one gift for the remainder of the chain.

\begin{lemma}[Chain patterns]\label{lem:chain-patterns}
	In round $k$, with $k - 1$ potential victims, the number of distinct chain patterns of length $\ell$ (i.e., $\ell$ steals followed by an open) is:
	\[
	\frac{(k-1)!}{(k-1-\ell)!}
	\]
	for $\ell \in \{1, 2, \ldots, k-1\}$.
\end{lemma}

\begin{proof}
	The first steal chooses from $k - 1$ victims. The second steal chooses from $k - 2$ (excluding the player whose gift is now locked). Continuing, the $\ell$-th steal chooses from $k - \ell$ options. The product is:
	\[
	(k-1)(k-2) \cdots (k-\ell) = \frac{(k-1)!}{(k-1-\ell)!}
	\]
\end{proof}

Let $A(k)$ denote the round action count, the number of distinct action patterns in round $k$, ignoring which specific wrapped gift is opened. Then:
\begin{equation}
	A(k) = 1 + \sum_{\ell=1}^{k-1} \frac{(k-1)!}{(k-1-\ell)!},
\end{equation}
where the leading $1$ accounts for the option to open immediately without stealing.

\begin{lemma}[Simplification]\label{lem:simplify}
	The round action count satisfies:
	\[
	A(k) = \sum_{j=0}^{k-1} \frac{(k-1)!}{j!}
	\]
	This is sequence \textup{A000522} in the OEIS, evaluated at $k - 1$.
\end{lemma}

\begin{proof}
	Simply substitute $j = k - 1 - \ell$ in the sum:
	\begin{align*}
		A(k) &= 1 + \sum_{\ell=1}^{k-1} \frac{(k-1)!}{(k-1-\ell)!} \\
		&= 1 + \sum_{j=0}^{k-2} \frac{(k-1)!}{j!} \\
		&= \frac{(k-1)!}{(k-1)!} + \sum_{j=0}^{k-2} \frac{(k-1)!}{j!} \\
		&= \sum_{j=0}^{k-1} \frac{(k-1)!}{j!}
	\end{align*}
\end{proof}

The first several values are:
\begin{center}
	\begin{tabular}{c|cccccccc}
		$k$ & 1 & 2 & 3 & 4 & 5 & 6 & 7 & 8 \\
		\hline
		$A(k)$ & 1 & 2 & 5 & 16 & 65 & 326 & 1957 & 13700
	\end{tabular}
\end{center}

\begin{remark}[Asymptotic growth]
	For large $m$, $\sum_{j=0}^{m} m!/j! \approx m! \cdot e$, where $e \approx 2.718$. 
	Thus $A(k) \sim (k-1)! \cdot e$ as $k \to \infty$.
\end{remark}

Now, let us get the count for a game of standard stealing rules.

\begin{theorem}[Trajectory count]\label{thm:trajectory-count}
	For a gift exchange game with $n$ players under standard stealing rules ($\ell_{\textup{round}} = 1$, $\ell_{\textup{total}} = \infty$), the number of distinct game trajectories (excluding the final swap) is:
	\begin{equation}
		T(n) = n! \times \prod_{k=1}^{n} A(k) = n! \times \prod_{k=0}^{n-1} \left( \sum_{j=0}^{k} \frac{k!}{j!} \right)
	\end{equation}
\end{theorem}

\begin{proof}
	In round $k$: there are $A(k)$ distinct action patterns (chain structures) and $n - k + 1$ choices for which wrapped gift to open. Since rounds proceed sequentially, and each round's gift choice is independent of the action pattern, we have that:
	\[
	T(n) = \prod_{k=1}^{n} \big[ A(k) \cdot (n - k + 1) \big]
	= \prod_{k=1}^{n} A(k) \times \prod_{k=1}^{n} (n - k + 1)
	= \prod_{k=1}^{n} A(k) \times n!
	\]
\end{proof}

Some explicit values are given in Table~\ref{tab:combin}.

\begin{table}[H]
	\centering
	\caption{Values for combinatorics}
	\label{tab:combin}
	\begin{tabular}{c|r|r}
		\toprule
		$n$ & $\prod_{k=1}^{n} A(k)$ & $T(n) = n! \times \prod A(k)$ \\ \midrule
		1 & 1 & 1 \\
		2 & 2 & 4 \\
		3 & 10 & 60 \\
		4 & 160 & 3,840 \\
		5 & 10,400 & 1,248,000 \\
		6 & 3,390,400 & 2,440,488,000 \\
		7 & 6,637,052,800 & $3.31 \times 10^{13}$ \\
		8 & $9.09 \times 10^{13}$ & $3.67 \times 10^{18}$ \\ \bottomrule
	\end{tabular}
\end{table}

Recall the tree structure description in Figure~\ref{fig:gametree}. Each path from root to an ``Open'' leaf represents one action pattern; the number of such paths depends on how many victims are available (which determines branching factor) and how chain-locking progressively reduces options (which bounds depth). The formula $A(k) = \sum_{j=0}^{k-1} \frac{(k-1)!}{j!}$ emerges from summing over all possible chain lengths, where each length $\ell$ contributes $(k-1)!/(k-1-\ell)!$ patterns corresponding to the ordered selection of $\ell$ distinct victims.

Let us take for example the case for $A(3)=5$, and list every action pattern in round 3\footnote{There is nothing especially remarkable about this; round 3 is chosen simply as it is the simplest case where nontrivial claims can form.}: we have exactly two potential victims  ($P_1$ and $P_2$), and the primary player $P_3$ can create chains of length 0, 1, or 2.

\begin{figure}[H]
	\centering
	\begin{tikzpicture}[
		% Node styles
		player/.style={
			circle,
			draw=black,
			line width=0.8pt,
			fill=blue!20,
			minimum size=0.8cm,
			font=\small
		},
		open/.style={
			rectangle,
			draw=black,
			line width=0.8pt,
			fill=green!25,
			rounded corners=3pt,
			minimum width=1.1cm,
			minimum height=0.6cm,
			font=\small\bfseries
		},
		steal/.style={
			rectangle,
			draw=black,
			line width=0.8pt,
			fill=red!25,
			rounded corners=3pt,
			minimum width=1.1cm,
			minimum height=0.6cm,
			font=\small\bfseries
		},
		arr/.style={->, >=Stealth, line width=0.8pt},
		patternlabel/.style={font=\normalsize\bfseries},
		chainlen/.style={font=\small, text=black!60}
		]
		
		\def\rowsep{-2.0}

		% ppen immediately (l = 0)
		
		\begin{scope}[shift={(0, 0)}]
			\node[patternlabel] at (-1.5, 0) {1.};
			
			\node[player] (p1_pk) at (0, 0) {$P_3$};
			\node[open] (p1_open) at (2.5, 0) {Open};
			
			\draw[arr] (p1_pk) -- (p1_open);
			
			\node[chainlen] at (5, 0) {$\ell = 0$};
			
			% Description
			\node[font=\small, text=black!70, anchor=west] at (6, 0) {$P_3$ opens a wrapped gift};
		\end{scope}

		% steal from P1, P1 opens (l = 1)
		
		\begin{scope}[shift={(0, \rowsep)}]
			\node[patternlabel] at (-1.5, 0) {2.};
			
			\node[player] (p2_pk) at (0, 0) {$P_3$};
			\node[steal] (p2_steal) at (2, 0) {Steal};
			\node[player] (p2_p1) at (4, 0) {$P_1$};
			\node[open] (p2_open) at (6, 0) {Open};
			
			\draw[arr] (p2_pk) -- (p2_steal);
			\draw[arr] (p2_steal) -- node[above, font=\scriptsize\itshape] {from} (p2_p1);
			\draw[arr] (p2_p1) -- (p2_open);
			
			\node[chainlen] at (8, 0) {$\ell = 1$};
		\end{scope}

		% steal from P2, P2 opens (l = 1)
		
		\begin{scope}[shift={(0, 2*\rowsep)}]
			\node[patternlabel] at (-1.5, 0) {3.};
			
			\node[player] (p3_pk) at (0, 0) {$P_3$};
			\node[steal] (p3_steal) at (2, 0) {Steal};
			\node[player] (p3_p2) at (4, 0) {$P_2$};
			\node[open] (p3_open) at (6, 0) {Open};
			
			\draw[arr] (p3_pk) -- (p3_steal);
			\draw[arr] (p3_steal) -- node[above, font=\scriptsize\itshape] {from} (p3_p2);
			\draw[arr] (p3_p2) -- (p3_open);
			
			\node[chainlen] at (8, 0) {$\ell = 1$};
		\end{scope}

		% P3 \toP1 to P2 \to Open (l = 2)
		
		\begin{scope}[shift={(0, 3*\rowsep)}]
			\node[patternlabel] at (-1.5, 0) {4.};
			
			\node[player] (p4_pk) at (0, 0) {$P_3$};
			\node[steal] (p4_s1) at (1.8, 0) {Steal};
			\node[player] (p4_p1) at (3.5, 0) {$P_1$};
			\node[steal] (p4_s2) at (5.3, 0) {Steal};
			\node[player] (p4_p2) at (7.1, 0) {$P_2$};
			\node[open] (p4_open) at (9, 0) {Open};
			
			\draw[arr] (p4_pk) -- (p4_s1);
			\draw[arr] (p4_s1) -- node[above, font=\scriptsize\itshape] {from} (p4_p1);
			\draw[arr] (p4_p1) -- (p4_s2);
			\draw[arr] (p4_s2) -- node[above, font=\scriptsize\itshape] {from} (p4_p2);
			\draw[arr] (p4_p2) -- (p4_open);
			
			\node[chainlen] at (11, 0) {$\ell = 2$};
		\end{scope}

		% P3 \to P2 \to P1 \to Open (l = 2)
		
		\begin{scope}[shift={(0, 4*\rowsep)}]
			\node[patternlabel] at (-1.5, 0) {5.};
			
			\node[player] (p5_pk) at (0, 0) {$P_3$};
			\node[steal] (p5_s1) at (1.8, 0) {Steal};
			\node[player] (p5_p2) at (3.5, 0) {$P_2$};
			\node[steal] (p5_s2) at (5.3, 0) {Steal};
			\node[player] (p5_p1) at (7.1, 0) {$P_1$};
			\node[open] (p5_open) at (9, 0) {Open};
			
			\draw[arr] (p5_pk) -- (p5_s1);
			\draw[arr] (p5_s1) -- node[above, font=\scriptsize\itshape] {from} (p5_p2);
			\draw[arr] (p5_p2) -- (p5_s2);
			\draw[arr] (p5_s2) -- node[above, font=\scriptsize\itshape] {from} (p5_p1);
			\draw[arr] (p5_p1) -- (p5_open);
			
			\node[chainlen] at (11, 0) {$\ell = 2$};
		\end{scope}
		
		\draw[decorate, decoration={brace, amplitude=8pt, mirror}, line width=1pt]
		(-2.2, 0.5) -- (-2.2, 4*\rowsep - 0.5)
		node[midway, left=10pt, align=center, font=\normalsize] {$A(3) = 5$\\patterns};
		
		\draw[gray, dashed, rounded corners=5pt] 
		(-1.8, \rowsep + 0.6) rectangle (8.5, 2*\rowsep - 0.6);
		\node[font=\scriptsize, text=gray, anchor=east] at (8.4, 1.5*\rowsep) {2 patterns};

		\draw[gray, dashed, rounded corners=5pt] 
		(-1.8, 3*\rowsep + 0.6) rectangle (11.5, 4*\rowsep - 0.6);
		\node[font=\scriptsize, text=gray, anchor=east] at (11.4, 3.5*\rowsep) {2 patterns};
		
		\begin{scope}[shift={(0, 4*\rowsep - 1.8)}]
			\node[font=\small\bfseries, anchor=west] at (0, 0) {Legend:};
			\node[player, minimum size=0.6cm] at (2, 0) {};
			\node[font=\small, anchor=west] at (2.5, 0) {Player};
			\node[steal, minimum width=0.8cm, minimum height=0.45cm] at (4.5, 0) {};
			\node[font=\small, anchor=west] at (5.1, 0) {Steal};
			\node[open, minimum width=0.8cm, minimum height=0.45cm] at (6.8, 0) {};
			\node[font=\small, anchor=west] at (7.4, 0) {Open};
		\end{scope}
		
	\end{tikzpicture}
	\caption{Chain-locking prevents cycles; once $P_3$ steals a gift, that gift cannot be stolen again within this chain, so $P_1$ cannot steal back from $P_3$ in patterns 4--5.}
	\label{fig:chainlock}
\end{figure}
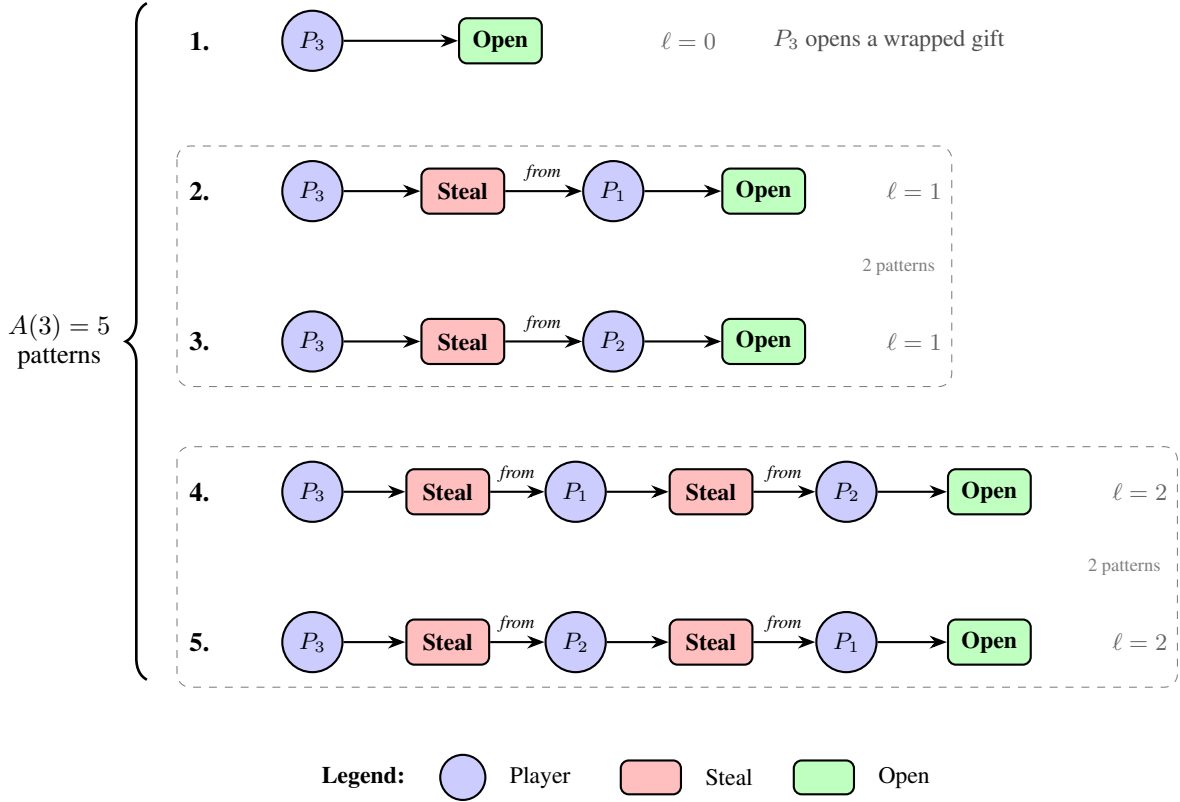

Figure~\ref{fig:chainlock} displays all five patterns, read left to right. Blue circles represent players making decisions; red ``steal'' boxes indicate a theft in progress; green ``open'' boxes mark chain termination. The chain length $\ell$ appears at the right of each row. Possible patterns are enumerated.

\begin{enumerate}
	\item with $\ell = 0$, $P_3$ opens immediately; no chain forms. This is the only way to have zero steals.
	
	\item[2--3.] with $\ell = 1$, $P_3$ steals from one victim, who then opens; there are exactly two such patterns because $P_3$ can steal from either $P_1$ or $P_2$.
	
	\item[4--5.] with $\ell = 2$, $P_3$ steals, the first victim  steals from the remaining player, and that player opens. Again, exactly two patterns exist: $P_3 \to P_1 \to P_2 \to \text{Open}$ and $P_3 \to P_2 \to P_1 \to \text{Open}$.
\end{enumerate}

The total count of $1 + 2 + 2 = 5$ matches the $A(3) = \sum_{j=0}^{2} 
\frac{2!}{j!} = 1 + 2 + 2 = 5$. Notably absent from the figure is any pattern where a victim steals \emph{back} from the player who just stole from them. Chain-locking prevents this: once $P_3$ takes a gift, that gift becomes temporarily unavailable within the current chain; this constraint is what makes the chain terminate in finite time, and is why the count is only $A(k)$ rather than something larger.

\begin{figure}[H]
	\centering
	\begin{tikzpicture}[
		scale=1.0,
		transform shape,
		player/.style={circle, draw=black, line width=0.6pt, fill=blue!20, 
			minimum size=0.65cm, font=\footnotesize},
		open/.style={rectangle, draw=black, line width=0.6pt, fill=green!25, 
			rounded corners=2pt, minimum width=1.1cm, minimum height=0.5cm, 
			font=\footnotesize\bfseries},
		steal/.style={rectangle, draw=black, line width=0.6pt, fill=red!25, 
			rounded corners=2pt, minimum width=0.85cm, minimum height=0.5cm, 
			font=\footnotesize\bfseries},
		outcome/.style={rectangle, draw=gray, line width=0.6pt, fill=gray!15, 
			rounded corners=2pt, minimum height=0.55cm, font=\footnotesize},
		trajlabel/.style={rectangle, draw=black, line width=0.8pt, fill=white, 
			minimum width=0.65cm, minimum height=0.5cm, font=\small\bfseries},
		arr/.style={->, >=Stealth, line width=0.6pt},
		outcomearr/.style={arr, color=gray}
		]
		
		\def\rowsep{-1.8}
		
		\node[font=\footnotesize\bfseries, text=gray] at (2.2, 0.8) {Round 1};
		\node[font=\footnotesize\bfseries, text=gray] at (5.5, 0.8) {Round 2};
		\node[font=\footnotesize\bfseries, text=gray] at (9.5, 0.8) {Outcome};
		
		% T1
		\begin{scope}[shift={(0, 0)}]
			\node[trajlabel] at (-0.4, 0) {T1};
			\node[player] (t1_p1) at (1, 0) {$P_1$};
			\node[open] (t1_o1) at (2.8, 0) {Open $G_1$};
			\node[player] (t1_p2) at (4.8, 0) {$P_2$};
			\node[open] (t1_o2) at (6.6, 0) {Open $G_2$};
			\node[outcome] (t1_out) at (9.5, 0) {$P_1\!:\!G_1,\,P_2\!:\!G_2$};
			\draw[arr] (t1_p1) -- (t1_o1);
			\draw[arr] (t1_o1) -- (t1_p2);
			\draw[arr] (t1_p2) -- (t1_o2);
			\draw[outcomearr] (t1_o2) -- (t1_out);
		\end{scope}
		
		% T2
		\begin{scope}[shift={(0, \rowsep)}]
			\node[trajlabel] at (-0.4, 0) {T2};
			\node[player] (t2_p1) at (1, 0) {$P_1$};
			\node[open] (t2_o1) at (2.8, 0) {Open $G_2$};
			\node[player] (t2_p2) at (4.8, 0) {$P_2$};
			\node[open] (t2_o2) at (6.6, 0) {Open $G_1$};
			\node[outcome] (t2_out) at (9.5, 0) {$P_1\!:\!G_2,\,P_2\!:\!G_1$};
			\draw[arr] (t2_p1) -- (t2_o1);
			\draw[arr] (t2_o1) -- (t2_p2);
			\draw[arr] (t2_p2) -- (t2_o2);
			\draw[outcomearr] (t2_o2) -- (t2_out);
		\end{scope}
		
		% T3
		\begin{scope}[shift={(0, 2*\rowsep)}]
			\node[trajlabel] at (-0.4, 0) {T3};
			\node[player] (t3_p1a) at (1, 0) {$P_1$};
			\node[open] (t3_o1) at (2.8, 0) {Open $G_1$};
			\node[player] (t3_p2) at (4.5, 0) {$P_2$};
			\node[steal] (t3_s) at (5.9, 0) {Steal};
			\node[player] (t3_p1b) at (7.2, 0) {$P_1$};
			\node[open] (t3_o2) at (8.8, 0) {Open $G_2$};
			\node[outcome] (t3_out) at (11.5, 0) {$P_1\!:\!G_2,\,P_2\!:\!G_1$};
			\draw[arr] (t3_p1a) -- (t3_o1);
			\draw[arr] (t3_o1) -- (t3_p2);
			\draw[arr] (t3_p2) -- (t3_s);
			\draw[arr] (t3_s) -- (t3_p1b);
			\draw[arr] (t3_p1b) -- (t3_o2);
			\draw[outcomearr] (t3_o2) -- (t3_out);
		\end{scope}
		
		% T4
		\begin{scope}[shift={(0, 3*\rowsep)}]
			\node[trajlabel] at (-0.4, 0) {T4};
			\node[player] (t4_p1a) at (1, 0) {$P_1$};
			\node[open] (t4_o1) at (2.8, 0) {Open $G_2$};
			\node[player] (t4_p2) at (4.5, 0) {$P_2$};
			\node[steal] (t4_s) at (5.9, 0) {Steal};
			\node[player] (t4_p1b) at (7.2, 0) {$P_1$};
			\node[open] (t4_o2) at (8.8, 0) {Open $G_1$};
			\node[outcome] (t4_out) at (11.5, 0) {$P_1\!:\!G_1,\,P_2\!:\!G_2$};
			\draw[arr] (t4_p1a) -- (t4_o1);
			\draw[arr] (t4_o1) -- (t4_p2);
			\draw[arr] (t4_p2) -- (t4_s);
			\draw[arr] (t4_s) -- (t4_p1b);
			\draw[arr] (t4_p1b) -- (t4_o2);
			\draw[outcomearr] (t4_o2) -- (t4_out);
		\end{scope}

		\begin{scope}[shift={(0, 4*\rowsep)}]
			\node[font=\small\bfseries, anchor=west] at (0, 0) {Legend:};
			\node[player, minimum size=0.6cm] at (2, 0) {};
			\node[font=\small, anchor=west] at (2.5, 0) {Player};
			\node[steal, minimum width=0.8cm, minimum height=0.45cm] at (4.5, 0) {};
			\node[font=\small, anchor=west] at (5.1, 0) {Steal};
			\node[open, minimum width=0.8cm, minimum height=0.45cm] at (6.8, 0) {};
			\node[font=\small, anchor=west] at (7.4, 0) {Open};
			\node[outcome, minimum width=0.8cm, minimum height=0.45cm] at (9.1, 0) {};
			\node[font=\small, anchor=west] at (9.7, 0) {Outcome};
		\end{scope}
		
	\end{tikzpicture}
	
	\caption{All four trajectories shown for a two-player game. Trajectories T1 and T4 reach the same final allocation ($P_1:G_1$, $P_2:G_2$) via different paths: T1 involves no stealing; T4 involves $P_2$ stealing then $P_1$ being forced to open. Similarly, T2 and T3 reach the same outcome via different paths, hence why we count \emph{trajectories} (paths) rather than \emph{outcomes} (allocations).}
	\label{fig:completegame}
\end{figure}
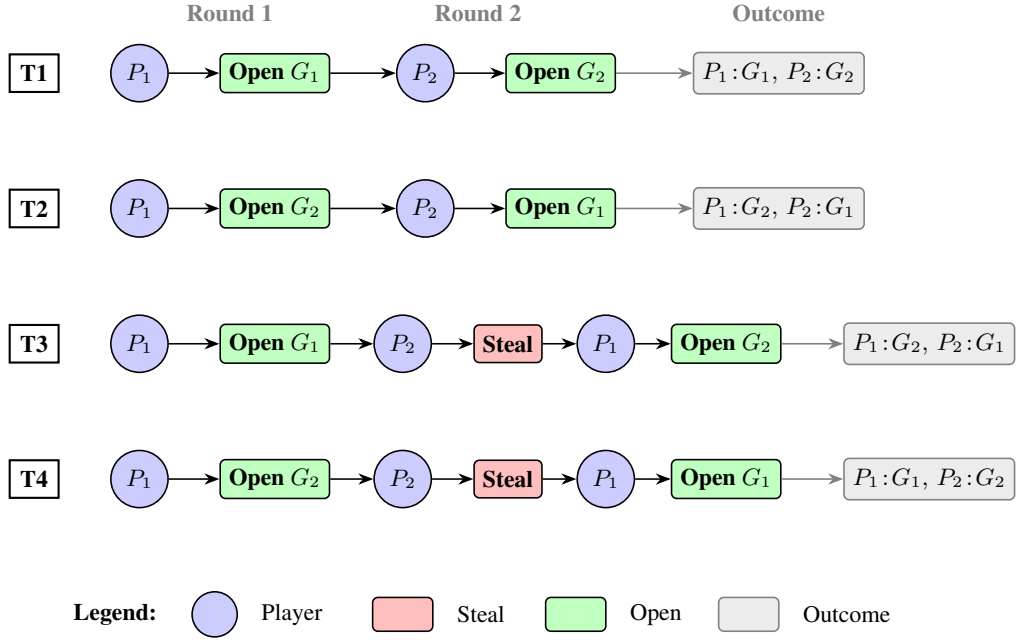

Figure~\ref{fig:completegame} enumerates every trajectory for the minimal interesting case: two players and two gifts. With only $T(2) = 4$ trajectories, we can display each one explicitly, and we observe that \emph{different trajectories (can) produce identical outcomes}, and hence the need to enumerate on trajectories instead of just outcomes.

Each row represents one trajectory, enumerated T1 through T4. Reading left to right, we see the sequence of actions across rounds. Gray boxes at the right show the final allocation of which player ends up with which gift.

We observe: trajectories T1 and T4 both end with the same allocation ($P_1$ gets $G_1$, $P_2$ gets $G_2$); similarly, T2 and T3 share an outcome. Yet T1 and T4 are genuinely different trajectories, and, hence, different games (or, gameplay experiences):
\begin{itemize}
	\item in \textbf{T1}, neither player steals; $P_1$ opens $G_1$, then $P_2$ opens $G_2$, and the game is calm and cooperative;
	
	\item in \textbf{T4}, $P_1$ opens $G_2$, but then $P_2$ steals $G_2$ from $P_1$, forcing $P_1$ to open the remaining gift $G_1$. The final allocation matches T1, but the social dynamics differ entirely!
\end{itemize}

It follows, then, that counting paths as opposed to outcomes is why the trajectory count $T(n)$ grows so much faster than the number of possible allocations $n!$. For two players, there are only $2! = 2$ distinct allocations, but there are $T(2) = 4$ ways to reach them. The excess factor exactly captures the drama between players: stealing chains; displacement; forced choices.

For larger games, this explodes. With $n = 8$ players, there are $8! \approx 40{,}000$ allocations but over $10^{18}$ trajectories. Almost all of the combinatorial complexity lies not in \emph{how} they get there.

Recall for the end game: the final swap gives Player 1 the option to exchange gifts with any other player, or to keep their current gift. This adds a factor of $n$ choices, and we recover:

\begin{equation}
	T_{\text{swap}}(n) = n \times T(n)
\end{equation}

We may consider now the limiting behaviors of locking. Under standard rules, chain-locking already constrains each gift to be stolen at most once per chain, and there is exactly one chain per round. Thus $\ell_{\text{round}} \geq 1$ has no effect on the count, since it is already enforced by chain-locking.

Setting $\ell_{\text{round}} = 0$ (unlimited) also has no effect, since chain-locking remains mandatory.

\begin{proposition}
	The trajectory count $T(n)$ is invariant to $\ell_{\textup{round}} \in \{0, 1, 2, \ldots\}$.
\end{proposition}

The lifetime limit $\ell_{\text{total}}$ \emph{does} affect the count. When $\ell_{\text{total}} < \infty$, gifts that reach the steal limit become permanently unstealable, reducing the number of valid targets in later rounds.

\begin{proposition}[Monotonicity]
	For fixed $n$, the trajectory count is weakly increasing in $\ell_{\textup{total}}$:
	\[
	\ell_{\textup{total}}' < \ell_{\textup{total}} \implies T(n; \ell_{\textup{total}}') \leq T(n; \ell_{\textup{total}})
	\]
	Equality holds when no gift would exceed the lower limit under any trajectory.
\end{proposition}

Computing $T(n; \ell_{\text{total}})$ exactly requires tracking the distribution of steal counts across gifts, which introduces path-dependence between rounds. We need not track \emph{which} gift has each steal count, only the \emph{multiset} of counts, since gifts with identical steal histories are interchangeable for counting purposes.

\subsection{Algorithm for finite total limit}

It is worth checking how we proceed with counting. We offload to coding this directly, and present an algorithm that computes the exact trajectory count for any $\ell_{\text{total}} \in \mathbb{Z}_{\geq 1} \cup \{\infty\}$.

\begin{algorithm}[H]
	\caption{Count distinct trajectories}
	\label{alg:trajectorycount}
	\begin{algorithmic}[1]
		\Require $n \in \mathbb{Z}_{\geq 1}$ (number of players)
		\Require $\ell_{\text{total}} \in \mathbb{Z}_{\geq 1} \cup \{\infty\}$ (lifetime steal limit)
		\Ensure $T$: total number of distinct game trajectories
		\Function{CountTrajectories}{$n, \ell_{\text{total}}$}
		\If{$\ell_{\text{total}} = \infty$}
		\State \Comment{closed-form derivation (see Theorem~\ref{thm:trajectory-count})}
		\State $P \gets 1$
		\For{$k \gets 0$ \textbf{to} $n-1$}
		\State $A_k \gets \sum_{j=0}^{k} \frac{k!}{j!}$ \Comment{A000522}
		\State $P \gets P \times A_k$
		\EndFor
		\State \Return $n! \times P$
		\Else
		\State \Comment{dynamic over steal-count multisets}
		\State let $S$ be a map: $\text{multiset} \to \mathbb{Z}_{\geq 0}$
		\State $S[\{0\}] \gets 1$ \Comment{after round 1: one gift opened; zero steals}
		
		\For{$r \gets 2$ \textbf{to} $n$} \Comment{rounds 2 through $n$}
		\State $S_{\text{next}} \gets \text{empty map}$
		\For{$(\mathbf{c}, w) \in S$} \Comment{$\mathbf{c}$: multiset of steal counts; $w$: ways}
		\State \Comment{action: open immediately (no stealing)}
		\State $\mathbf{c}_{\text{open}} \gets \mathbf{c} \cup \{0\}$
		\State $S_{\text{next}}[\mathbf{c}_{\text{open}}] \mathrel{+}= w$
		
		\State \Comment{action: initiate stealing chain}
		\State $G_{\text{valid}} \gets \{g \in \mathbf{c} : g < \ell_{\text{total}}\}$
		\If{$G_{\text{valid}} \neq \emptyset$}
		\State $W_{\text{chains}} \gets \Call{CountChains}{\mathbf{c}, G_{\text{valid}}}$
		\For{$(\mathbf{c}_{\text{final}}, w_{\text{chain}}) \in W_{\text{chains}}$}
		\State $\mathbf{c}_{\text{new}} \gets \mathbf{c}_{\text{final}} \cup \{0\}$ \Comment{chain ends with open}
		\State $S_{\text{next}}[\mathbf{c}_{\text{new}}] \mathrel{+}= w \times w_{\text{chain}}$
		\EndFor
		\EndIf
		\EndFor
		\State $S \gets S_{\text{next}}$
		\EndFor
		\State \Return $n! \times \sum_{\mathbf{c}} S[\mathbf{c}]$
		\EndIf
		\EndFunction
		
		\Statex
		
		\Function{CountChains}{$\mathbf{c}_{\text{start}}, \text{targets}$}
		\State \Comment{recursively enumerate all valid stealing chains}
		\State \Comment{returns final multiset $\to$ number of chains}
		\State $R \gets \text{empty map}$
		\For{$g \in \text{targets}$}
		\State $\mathbf{c}' \gets \mathbf{c}_{\text{start}}$ with one instance of $g$ incremented to $g+1$
		\State $R[\mathbf{c}'] \mathrel{+}= 1$ \Comment{chain ends here (victim opens)}
		
		\State $\text{remaining} \gets \text{targets} \setminus \{g\}$ \Comment{chain-lock: $g$ unavailable}
		\State $R_{\text{sub}} \gets \Call{CountChains}{\mathbf{c}', \text{remaining}}$
		\For{$(\mathbf{c}_{\text{end}}, w) \in R_{\text{sub}}$}
		\State $R[\mathbf{c}_{\text{end}}] \mathrel{+}= w$ \Comment{accumulate counts}
		\EndFor
		\EndFor
		\State \Return $R$
		\EndFunction
	\end{algorithmic}
\end{algorithm}

Figure~\ref{fig:dynamic} illustrates how the trajectory-counting algorithm operates when a lifetime steal limit $\ell_{\text{total}}$ is in effect. We need only track how many gifts have each steal count; gifts with identical histories are interchangeable for counting purposes. Figure~\ref{fig:dynamic} displays state transitions across four rounds, read top to bottom. Round 1 begins with a single state $\{0\}$: one gift has been opened and never stolen. Two types of transitions connect states across rounds:

\begin{itemize}
	\item for green arrows (Open), the primary player opens a new gift without stealing; this adds a 0 to the multiset, representing a fresh gift entering play., and the transition $\{0\} \to \{0, 0\}$ shows this as one unstolen gift becomes two.
	
	\item for red arrows (Steal), the primary player initiates a stealing chain, which increments the steal count of one or more gifts before terminating with an open; the transition $\{0\} \to \{0, 1\}$ shows a single steal: the existing gift's count increments from 0 to 1, and a new gift (count 0) enters when the chain terminates.
\end{itemize}

The dashed red arrow from $\{0, 1\}$ to $\{0, 0, 2\}$ represents a length-2 chain where the same gift is stolen twice in succession; this produces a gift with steal count 2, shown in the orange node. If the lifetime limit were $\ell_{\text{total}} = 2$, for example, this gift would technically be ``frozen'', permanently removed from the pool of stealable targets.

\begin{figure}[H]
	\centering
	\begin{tikzpicture}[
		% Node styles
		state/.style={
			rectangle,
			draw=black,
			line width=1pt,
			fill=blue!15,
			rounded corners=5pt,
			minimum width=1.8cm,minimum height=0.8cm,font=\normalsize\ttfamily
		},
		frozen/.style={
			state,
			fill=orange!30
		},
		roundlabel/.style={
			font=\normalsize\bfseries,
			text=black
		},
		openarr/.style={
			->,
			>=Stealth,
			line width=1.5pt,
			color=green!60!black
		},
		stealarr/.style={
			->,
			>=Stealth,
			line width=1.5pt,
			color=red!70!black
		},
		stealdasharr/.style={
			stealarr,
			dashed
		},
		grayarr/.style={
			->,
			>=Stealth,
			line width=1pt,
			color=gray!50
		},
		arrlabel/.style={
			font=\small\itshape
		}
		]
		
		\node[roundlabel] at (-0.8, 0) {R1};
		\node[roundlabel] at (-0.8, -3) {R2};
		\node[roundlabel] at (-0.8, -6.5) {R3};
		\node[roundlabel] at (-0.8, -10) {R4};
		
		%round 1
		\node[state] (r1) at (4, 0) {\{0\}};
		\node[font=\small, text=gray, above=0.2cm of r1] {1 way};
		
		%round 2
		\node[state] (r2a) at (1.5, -3) {\{0,0\}};
		\node[state] (r2b) at (6.5, -3) {\{0,1\}};
		
		% arrows from round 1
		\draw[openarr] (r1) -- (r2a) 
		node[midway, above left, arrlabel, color=green!60!black] {open};
		\draw[stealarr] (r1) -- (r2b) 
		node[midway, above right, arrlabel, color=red!70!black] {steal};
		
		%round 3
		\node[state] (r3a) at (0.8, -6.5) {\{0,0,0\}};
		\node[state] (r3b) at (4, -6.5) {\{0,0,1\}};
		\node[state] (r3c) at (7.2, -6.5) {\{0,1,1\}};
		\node[frozen] (r3d) at (10.4, -6.5) {\{0,0,2\}};
		
		\draw[openarr] (r2a) -- (r3a);
		\draw[stealarr] (r2a) -- (r3b);
		\draw[openarr] (r2b) -- (r3b);
		\draw[stealarr] (r2b) -- (r3c);
		\draw[stealdasharr] (r2b) to[bend right=15] (r3d);
		
		%round 4
		\node[font=\large\itshape, text=gray] (r4) at (5.5, -10) {\ldots continue \ldots};
		
		%  arrows to Round 4
		\draw[grayarr] (r3a) -- (r4);
		\draw[grayarr] (r3b) -- (r4);
		\draw[grayarr] (r3c) -- (r4);
		\draw[grayarr] (r3d) -- (r4);
		
		\begin{scope}[shift={(10, -0.5)}]
			\node[font=\normalsize\bfseries, anchor=west] at (0, 0) {Legend};
			
			\draw[openarr] (0, -0.8) -- (1.2, -0.8);
			\node[font=\small, anchor=west] at (1.5, -0.8) {Open (add 0)};
			\draw[stealarr] (0, -1.5) -- (1.2, -1.5);
			\node[font=\small, anchor=west] at (1.5, -1.5) {Steal (increment)};
			\node[frozen, minimum width=1.2cm, minimum height=0.5cm] at (0.6, -2.3) {};
			\node[font=\small, anchor=west] at (1.5, -2.3) {Gift stolen twice};
		\end{scope}

	\end{tikzpicture}
	\caption{Each node in the diagram represents a \emph{state} after the corresponding round completes. States are written as multisets of steal counts using set notation with possible repetition. $\{0, 0, 1\}$ indicates that three gifts have been opened: two have never been stolen (count 0), and one has been stolen exactly once (count 1). The multiset representation abstracts away gift identity, capturing only the distribution of steal counts. With $\ell_{\text{total}} = 2$, the \texttt{\{0,0,2\}} state has one ``frozen'' gift that cannot be stolen again.}
	\label{fig:dynamic}
\end{figure}
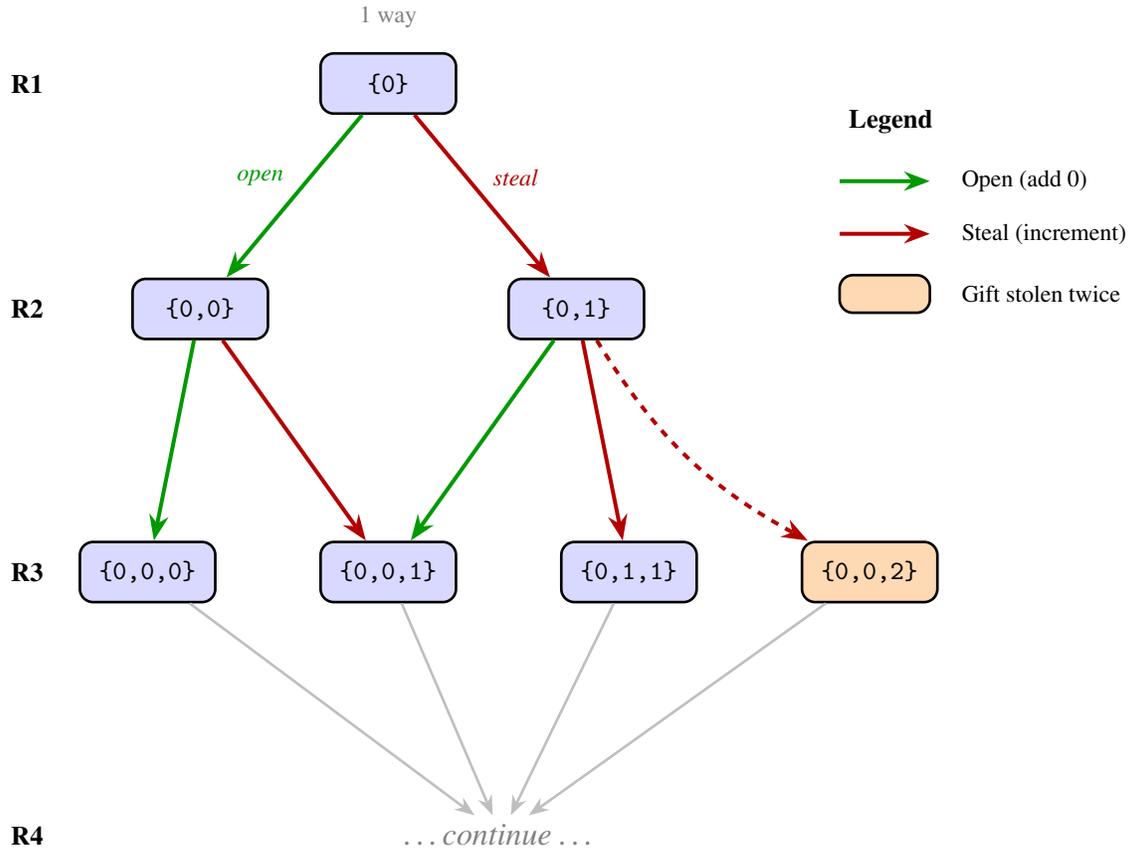

We process round by round, maintaining a count of how many trajectory prefixes lead to each state. At each transition, we enumerate all valid chains (respecting $\ell_{\text{total}}$), update the destination states, and accumulate counts. After round $n$, sums across all final states and multiplies by $n!$ to recover gift identity and our total trajectory count of $T(n; \ell_{\text{total}})$.

Note that it is exactly the multiset representation that makes this tractable; else we need to track the full history of each individual gift! By exploiting exchangeability, the number of distinct states grows polynomially in $n$. We also observe why steal limits reduce trajectory counts. Compare the branching from $\{0, 1\}$: without limits, both the once-stolen gift and the unstolen gift are valid targets, producing much branching; with$\ell_{\text{total}} = 2$, the twice-stolen gift in $\{0, 0, 2\}$ becomes frozen, reducing future options. This is essentially a pruning oeration that compounds across rounds, explaining why $T(n; 3)$ is orders of magnitude smaller than $T(n; \infty)$.

\begin{remark}[State space]
	The complexity depends on the number of distinct multisets of steal counts. For $\ell_{\text{total}} = 1$, each gift has count in $\{0, 1\}$, yielding $O(n)$ distinct multisets; for larger $\ell_{\text{total}}$, the state space grows, certainly, but remains polynomial in $n$ for fixed $\ell_{\text{total}}$.
\end{remark}

\begin{remark}[Multiplication by $n!$]
	The final multiplication by $n!$ accounts for which specific wrapped gift is opened in each round. The algorithm counts action patterns (who steals from whom); gift identity needs must be tracked separately.
\end{remark}

The rapid growth of $T(n)$ is exactly the combinatorial explosion of possible game outcomes. Even for modest $n$ we observe:
\begin{itemize}
	\item for $n = 6$, over 2 billion distinct trajectories;
	\item for $n = 8$, over $10^{18}$ trajectories;
	\item for $n = 10$, exceeds $10^{26}$ trajectories.
\end{itemize}

This explains why repeated plays of the gift exchange game feel different: the probability of replaying an identical trajectory is remarkably small!

\section{Future work}\label{sec:future}

We outline directions of future work.

\subsection{Coalitions, collusion}

In practice, players may coordinate informally\footnote{Anecdotally, it is rare to \emph{not} form some kind of coalition!}. We sketch a formalization of such behavior.

\begin{definition}[Coalition]
	A \emph{coalition} is a nonempty subset $K \subseteq P$ of players who coordinate their strategies. The \emph{coalition structure} is a partition $\mathcal{K} = \{K_1, K_2, \ldots, K_m\}$ of $P$ into disjoint coalitions.
\end{definition}

\begin{definition}[Singleton structure]
	The \emph{singleton structure} $\mathcal{K}_0 = \{\{P_1\}, \{P_2\}, \ldots, \{P_n\}\}$ represents fully non-cooperative play. 
\end{definition}

Our base analysis assumes $\mathcal{K}_0$, so the coalition formalization is technically more general than what we have already worked out here.

\begin{definition}[Grand coalition]
	The \emph{grand coalition} $\mathcal{K}_* = \{P\}$ represents full cooperation. Under $\mathcal{K}_*$, players jointly maximize total welfare, which is trivial in the gift exchange game: assign each gift to the player who values it most.
\end{definition}

A coalition $K$ may pursue several objectives:

\begin{definition}[Coalition welfare]
	The \emph{welfare} of coalition $K$ at terminal state $S_n$ is:
	\[
	W_K = \sum_{P_i \in K} V_i(O(P_i))
	\]
\end{definition}

Accordingly, coalitions may maximize $W_K$ rather than individual utilities.

\begin{definition}[Egalitarian objective]
	Alternatively, $K$ may maximize the minimum member outcome:
	\[
	W_K^{\min} = \min_{P_i \in K} V_i(O(P_i))
	\]
\end{definition}

In the egalitarian objective, we model family coalitions where ``everyone should get something decent''. Now, what counts as ``decent'' is certainly idiosyncratic, but we may apply the strategies we have covered so far in that regard.

Coalitions may implement several coordination mechanisms. We suggest a few for starters.

\begin{definition}[Non-aggression pact]
	Coalition $K$ follows a \emph{non-aggression pact} if:
	\[
	\forall P_i, P_j \in K: P_i \text{ never executes } \textsc{Steal}(P_j)
	\]
	
\end{definition}

The non-aggression pact reduces the effective steal targets for coalition members to $P \setminus K$.

\begin{definition}[Coordinated targeting]
	Coalition $K$ follows \emph{coordinated targeting} if members preferentially steal from non-members:
	\[
	\forall P_i \in K: \mathcal{T}_K(P_i, S_k) = \{P_m \in \mathcal{T}(P_i, S_k) : P_m \notin K\}
	\]
	when this set is nonempty.
\end{definition}

\begin{definition}[Gift laundering]
	A \emph{laundering chain} occurs when $P_i \in K$ steals a gift $G$ from $P_m \notin K$, intending for coalition partner $P_j \in K$ (who is later in turn order) to steal $G$ from $P_i$. This transfers $G$ to $P_j$ while ``spending'' $P_i$'s steal opportunity.
\end{definition}

\begin{definition}[Information sharing]
	Under partial information, coalition members may share private signals:
	\[
	\mathcal{I}_K(P_i) = \bigcup_{P_j \in K} \mathcal{I}(P_j)
	\]
\end{definition}

Information sharing is valuable: it expands the information set for coalition members, improving gift selection at all.

We might now wonder at the stability of allegiences. 

\begin{definition}[Individual rationality]
	Coalition $K$ is \emph{individually rational} for player $P_i$ if:
	\[
	\mathbb{E}[V_i(O(P_i)) \mid P_i \in K] \geq \mathbb{E}[V_i(O(P_i)) \mid \mathcal{K}_0]
	\]
	
\end{definition}

That is, membership weakly improves $P_i$'s expected outcome versus singleton play.

\begin{definition}[Core stability]
	A coalition structure $\mathcal{K}$ is \emph{core stable} if no subset $S \subseteq P$ can deviate to form a new coalition and make all members of $S$ strictly better off.
\end{definition}

\begin{conjecture}[Non-empty core]
	Under standard gift exchange rules with correlated valuations, the coalition formation game has a non-empty core. In particular, family-sized coalitions ($|K| \in \{2, 3, 4\}$) among players with correlated preferences are stable.
\end{conjecture}

In empirical settings, coalitions may not be announced, or even found out! Observable signatures include:

\begin{itemize}
	\item asymmetric stealing rates may be observed as $H_{ij} \ll H_{ik}$ for $j \in K_i$, $k \notin K_i$;
	\item laundering patterns follow as gift paths traverse coalition members;
	\item information leakage is most damning, as coalition members select better wrapped gifts.
\end{itemize}

This lends itself nicely to a graph structure. We only define here, and leave for future work to explore this formalization.

\begin{definition}[Coalition graph]
	The \emph{empirical coalition graph} $G_K = (P, E_K)$ has edge $(P_i, P_j) \in E_K$ if the observed steal rate $H_{ij}/n_i^{\text{steal}}$ is significantly below the null expectation under singleton play.
\end{definition}

\subsection{The Own-gift prohibition}

A common, but frequently unenforced, rule prohibits players from choosing and/or ending up with the gift they brought. In many groups, players are expected to avoid this; the prohibition may take any of several forms:

\begin{enumerate}
	\item a player \emph{cannot} open or steal their own gift; if displaced and only their own gift remains stealable, they must open from the wrapped pile;
	
	\item a player \emph{should not} take their own gift, but the rules do not prevent it. Violation incurs social disapproval, but no real mechanical consequence;
	
	\item the rule exists in principle but is rarely invoked, especially when gifts are anonymous or players forget who brought what in the first place.
\end{enumerate}

In practice, players occasionally end up with their own gift through complex stealing chains or late-game necessity, prompting mild embarrassment or entertainment, but not disqualification.

Let us formalize. Let $B: P \to G$ denote the \emph{brought-by} function, where $B(P_i) = G_j$ indicates that player $P_i$ contributed gift $G_j$.

\begin{definition}[Own-gift constraint]
	The \emph{hard own-gift constraint} modifies valid actions as follows:
	\begin{itemize}
		\item $\textsc{Open}(G_j)$ is invalid for $P_i$ if $B(P_i) = G_j$;
		\item $\textsc{Steal}(P_m)$ is invalid for $P_i$ if $B(P_i) = O(P_m)$.
	\end{itemize}
\end{definition}

\begin{definition}[Own-gift penalty]
	The \emph{soft own-gift constraint} is a penalty that adds a social cost term:
	\[
	c_{\text{own}}(P_i) = 
	\begin{cases}
		\kappa & \text{if } O(P_i) = B(P_i) \text{ at game end} \\
		0 & \text{otherwise}
	\end{cases}
	\]
	where $\kappa > 0$ represents a measure of embarrassment or public shame of ``getting your own gift back''.
\end{definition}

The own-gift constraint interacts with other game mechanics. Under the hard constraint, each player has one fewer valid steal target and (potentially) one fewer valid open choice; this slightly increases the probability of being forced to open rather than steal. If players know who brought which gift, the own-gift constraint reveals information for free: observing that $P_i$ \emph{could} steal $G_j$ but does not may signal that $B(P_i) = G_j$, or, simply, that $P_i$ does not want $G_j$. With the hard constraint and small $n$, edge cases arise, in that if a displaced player's only valid steal target is their own gift and no wrapped gifts remain, the game reaches an impasse, and rules should allow in this case an own-gift acquisition in such cases. Finally, if players anticipate the constraint, they may strategically bring gifts they do not want, knowing they cannot end up with them, a humorous (or, perhaps, machiavellian) inversion to ``bring something you would want yourself''.

We omitted the own-gift constraint from our primary analysis for three reasons:

\begin{enumerate}
	\item the rule ranges from hard constraint to ignored entirely across groups, making it difficult to model universally;
	
	\item the constraint requires tracking $B(P_i)$ for all players, adding state complexity; in anonymous exchanges, this information may not exist;
	
	\item in simulations with $n \geq 8$, the probability of ending up with one's own gift by chance is low enough ($\approx 1/n$ under random assignment). 
\end{enumerate}

A complete treatment should parameterize enforcement ($\kappa = 0$ for unenforced, $\kappa = \infty$ for hard constraint) and analyze how the constraint affects equilibrium play; we leave this for future work.

\subsection{Mechanism design}

The current analysis takes the gift exchange game mechanism as given. A natural extension asks: what mechanisms optimize given objectives?

\begin{question}[Efficient mechanism]
	Does there exist a stealing-based gift exchange mechanism that achieves Pareto efficiency \cite{arunachaleswaran2024pareto} (each gift goes to the player who values it most) in expectation, without requiring truthful revelation of preferences?
\end{question}

\begin{question}[Envy minimization]
	Can we design a mechanism that minimizes ex-post envy \cite{feldman2024breaking}, defined as:
	\[
	\text{Envy} = \sum_{i} \max_{j \neq i} \big( V_i(O(P_j)) - V_i(O(P_i)) \big)^+
	\]
	The final swap partially addresses this for Player 1; can it be generalized?
\end{question}

\begin{question}[Optimal Stealing Limits]
	Given $n$ players with valuation model $\mathcal{M}$, what values of $(\ell_{\text{round}}, \ell_{\text{total}})$ maximize social welfare? Minimize game duration? Balance efficiency and entertainment?
\end{question}

\subsection{Empirical validation}

The simulation results rely on assumed behavioral parameters. Validation requires: controrlled gift exchanges with induced valuations, measuring actual steal rates, chain lengths, and strategy choices; systamatic observational data from real games, with post-game surveys eliciting valuations and satisfaction; structural estimation of social cost parameters $(c_0, \alpha, \beta)$ and adaptive strategy coefficients $(\lambda_1, \lambda_2, \lambda_3)$ from choice data.

\begin{question}[Behavioral validity]
	Do the logit choice model and linear social cost function adequately describe human behavior, or are threshold effects, reference dependence, or other behavioral phenomena empirically significant?
\end{question}

\subsection{Learning and repeated play}

Some groups play the gift exchange game more frequently than others; this introduces dynamics:

\begin{definition}[Reputation across games]
	Let $H_i^{(t)}$ denote player $P_i$'s stealing history through game $t$. A \emph{reputation function} $\rho_i^{(t)} = f(H_i^{(1)}, \ldots, H_i^{(t)})$ summarizes how others perceive $P_i$ entering game $t + 1$.
\end{definition}

\begin{question}[Reputation equilibrium]
	In a repeated game with stable player set, does play converge to an equilibrium? Do aggressive players develop bad reputations that constrain future behavior?
\end{question}

\begin{question}[Strategy learning]
	If players update strategies based on outcomes, what learning dynamics emerge? Do groups converge toward cooperative norms, aggressive escalation, or cycling?
\end{question}

\subsection{Incomplete information about player types}

We assumed players know each other's utility functions (or at least their strategic tendencies). Relaxing this induces a space.

\begin{definition}[Type space]
	Each player $P_i$ has a \emph{type} $\theta_i \in \Theta$ that decorates valuation, drawn from a common prior $F(\theta)$. Types determine valuations and strategic behavior:
	\[
	V_i = V(\cdot; \theta_i), \quad \sigma_i = \sigma(\cdot; \theta_i)
	\]
\end{definition}

\begin{question}[Bayesian equilibrium]
	Under type uncertainty, what constitutes a Bayesian equilibrium in the gift exchange game? How do players update beliefs about opponents' types from observed actions?
\end{question}

\section{Conclusion}\label{sec:conclusion}

We have provisions a formalization of the gift exchange game, proving basic properties (chain termination, end-state bijection) and deriving trajectory counts that grow superexponentially in player number. The decorated game framework is parameterized by partial information, social costs, adaptive dynamics, and biased selection, and permits controlled analysis of behavioral factors typically ignored in strategic models.

Social costs are the primary regulator of stealing behavior, reducing theft rates by 27--48\% depending on valuation structure; implicit social friction in real exchanges substantially moderates the aggressive play that purely strategic analysis would predict. Adaptive dynamics favor passive strategies: satisfaction dampening and frustration feedback compress the performance gap between aggressive and passive play, consistent with ecological rationality perspectives where simple heuristics outperform optimization in socially embedded contexts. Value correlation amplifies all behavioral effects; when consensus exists about gift quality, competition intensifies and the scope for strategic differentiation expands.

Partial information, despite its theoretical grounding in Bayesian belief updating and CARA risk adjustment, contributes minimally to observed dynamics. The asymmetry between uncertain opening and certain stealing slightly favors theft, opposite to our initial hypothesis. This null result is itself informative: uncertainty about wrapped gifts is not the primary driver of strategic complexity in these games.

The positional analysis confirms intuition: first-player advantage is insurmountable (the final swap guarantees Seat 1 can claim the best gift), while Seat 2 occupies the worst position across all configurations; no behavioral feature reverses these orderings.

Several limitations bound interpretation. The linear social cost function and logit choice model are tractable approximations; threshold effects, reference dependence, and richer affective dynamics remain unexplored. Strategy assignment was uniform random; real populations exhibit correlated types and learning. The simulation parameters were calibrated for plausible behavior, and are not estimated from choice data. Empirical validation, of which includes (but certainly not limited to) controlled exchanges with induced valuations, or structural estimation from observational data, remains necessary.

We hope that the formalization explored here opens directions we have only sketched: coalition formation among family subgroups; mechanism design for efficiency or envy minimization; learning dynamics in repeated play; Bayesian equilibrium under type uncertainty. The gift exchange game, for all its holiday frivolity, provides a compact laboratory for studying competition, social preferences, and strategic behavior under uncertainty.

Finally, let fun be not the least of things. The gift exchange game is a social game, where the primary objective is to have fun. Strategic elements should never supplant the game's fundamental purpose of creating an entertaining gift-giving experience. Formalities here should always come a distant second to having fun with friends and family. Payers should feel free to make suboptimal choices that enhance the social experience. It is not that these are all mutually exclusive; understanding the game's structure may encourage an appreciation and spark interesting discussions, but if competition compromises fun, then the former should be relaxed for the sake of latter. Have fun; please do not be \emph{that} guy.

\newpage
%\printbibliography
\bibliography{references}

\appendix

\section{Notation reference}

Here we tabulate the notation and parameters used throughout our work here.

\begin{table}[H]
	\centering
	\caption{Summary of notation used throughout}
	\begin{tabular}{cl}
		\toprule
		\bfseries Symbol & \bfseries Meaning \\
		\midrule
		$P = \{P_1, \ldots, P_n\}$ & player set \\
		$G = \{G_1, \ldots, G_n\}$ & gift set \\
		$\owners: P \rightarrow G \cup \{\bot\}$ & ownership function \\
		$\status: G \rightarrow \{\textsc{wrapped}, \textsc{opened}\}$ & gift status function \\
		$\locked_k \subseteq G$ & locked gifts in round $k$ \\
		$V_i: G \rightarrow [0,1]$ & player $i$'s valuation function \\
		$\valmat \in [0,1]^{n \times n}$ & valuation matrix \\
		$q_j \in [0,1]$ & objective quality of gift $j$ \\
		$a_j \in [0,1]$ & appearance signal of gift $j$ \\
		$\hat{V}_i(G_j)$ & perceived value under partial information \\
		$c(P_i, P_m)$ & social cost of $P_i$ stealing from $P_m$ \\
		$\phi_i \in [0,1]$ & frustration level of player $i$ \\
		$U_i^{\text{steal}}(P_m)$ & net utility of $P_i$ stealing from $P_m$ \\
		$\Features \subseteq \{\PI, \SC, \AD, \BS\}$ & enabled feature set \\
		$\mathcal{G}_b$ & base game ($\Features = \emptyset$) \\
		$\mathcal{G}_d$ & decorated game \\
		$\Delta_f(Y)$ & marginal effect of feature $f$ on outcome $Y$ \\
		$I_{f_1, f_2}(Y)$ & pairwise interaction effect \\
		\bottomrule
	\end{tabular}
	
\end{table}

\begin{table}[h]
	\centering
	\caption{Model parameters by feature}
	\begin{tabular}{cll}
		\toprule
		\bfseries Parameter & \bfseries Feature & \bfseries Description \\
		\midrule
		$\rho$ & $\mathcal{M}_{\text{cor}}$ & correlation strength \\
		$\sigma$ & $\mathcal{M}_{\text{neg}}$ & noise standard deviation \\
		$\sigma_a$ & $\PI$ & appearance signal noise \\
		$\mu_0$ & $\PI$ & prior mean for wrapped gifts \\
		$\delta$ & $\PI$ & uncertainty penalty \\
		$c_0$ & $\SC$ & base stealing cost \\
		$\alpha$ & $\SC$ & repeat-offense multiplier \\
		$\beta$ & $\SC$ & reputation decay rate \\
		$\gamma$ & $\SC$ & frustration gain per theft \\
		$\gamma'$ & $\SC$ & frustration decay rate \\
		$\lambda_1$ & $\AD$ & phase aggression coefficient \\
		$\lambda_2$ & $\AD$ & frustration aggression coefficient \\
		$\lambda_3$ & $\AD$ & satisfaction dampening coefficient \\
		$\tau$ & $\BS$ & selection temperature \\
		\bottomrule
	\end{tabular}
\end{table}

\section{Complexity}

We report the complexities.

\subsection{Time complexity}

The time complexity of a single player decision is $O(n)$:

\begin{itemize}
	\item computing valid steal targets is $O(n)$;
	\item evaluating net utility for each target is $O(n) \times O(1) = O(n)$;
	\item finding maximum is $O(n)$.
\end{itemize}

The worst-case time complexity of a stealing chain is $O(n^2)$:

\begin{itemize}
	\item maximum chain length is $O(n)$ (follows from Lemma~\ref{lem:termination});
	\item decisions per chain $O(n)$;
	\item cost per decision $O(n)$
\end{itemize}

The time complexity of a single game is $O(n^3)$:
\begin{itemize}
	\item $n$ rounds;
	\item each round contributes $O(n^2)$ (one chain);
	\item total follows: $O(n \cdot n^2) = O(n^3)$.
\end{itemize}

For $g$ games with $n$ players, total simulation complexity is $O(gn^3)$.

\subsection{Space complexity}

Per-game space complexity is $O(n^2)$:
\begin{itemize}
	\item valuation matrix $\valmat$ $O(n^2)$;
	\item social state (all players) is $O(n \times n) = O(n^2)$ (essentially history dictionaries);
	\item game state is $O(n)$
\end{itemize}

\subsection{Base and decorated game}

$\mathcal{G}_b$ and $\mathcal{G}_d$ have identical asymptotic complexity; decorations add constant-factor overhead:
\begin{itemize}
	\item $\PI$ gives $O(1)$ per valuation lookup;
	\item $\SC$ gives $O(1)$ per steal (dictionary operations);
	\item $\AD$ gives $O(n)$ per decision (phase calculation);
	\item $\BS$ gives $O(n)$ per open (softmax computation).
\end{itemize}

Total overhead: approximately $2$--$3\times$ in practice.

\section{Tabulated results}

For completeness, we report the tabulated results of our simulation experiment.

\begin{table}[H]
	\centering
	\caption{Feature effect comparison (independent valuation)}
	\begin{tabular}{lccccc}
		\toprule
		\bfseries Configuration & \bfseries Steals/Game & \bfseries Chain length & \bfseries Seat 1 & \bfseries Seat 2 & \bfseries Seat 29 \\
		\midrule
		BASE & 57.017 & 2.993 & 0.966 & 0.757 & 0.814 \\
		PI & 57.017 & 2.993 & 0.966 & 0.757 & 0.814 \\
		SC & 46.359 & 2.443 & 0.934 & 0.730 & 0.819 \\
		AD & 54.524 & 3.189 & 0.967 & 0.774 & 0.815 \\
		BS & 58.311 & 3.050 & 0.967 & 0.763 & 0.813 \\
		PI+SC & 46.359 & 2.443 & 0.934 & 0.730 & 0.819 \\
		PI+AD & 54.524 & 3.189 & 0.967 & 0.774 & 0.815 \\
		PI+BS & 58.311 & 3.050 & 0.967 & 0.763 & 0.813 \\
		SC+AD & 53.103 & 3.115 & 0.929 & 0.777 & 0.816 \\
		SC+BS & 46.395 & 2.444 & 0.935 & 0.713 & 0.812 \\
		AD+BS & 53.762 & 3.171 & 0.967 & 0.770 & 0.804 \\
		PI+SC+AD & 53.103 & 3.115 & 0.929 & 0.777 & 0.816 \\
		PI+SC+BS & 46.395 & 2.444 & 0.935 & 0.713 & 0.812 \\
		PI+AD+BS & 53.762 & 3.171 & 0.967 & 0.770 & 0.804 \\
		SC+AD+BS & 52.921 & 3.104 & 0.928 & 0.771 & 0.821 \\
		FULL & 52.921 & 3.104 & 0.928 & 0.771 & 0.821 \\
		\bottomrule
	\end{tabular}
\end{table}

\begin{table}[H]
	\centering
	\caption{Strategy performance by configuration (independent valuation)}
	\begin{tabular}{lccccc}
		\toprule
		\bfseries Config & \texttt{always\_open} & \texttt{always\_steal} & \texttt{coin\_flip} & \texttt{mean\_based} & \texttt{threshold} \\
		\midrule
		BASE & 0.5160 & 0.9129 & 0.7174 & 0.9147 & 0.9141 \\
		PI & 0.5160 & 0.9129 & 0.7174 & 0.9147 & 0.9141 \\
		SC & 0.5136 & 0.9165 & 0.7244 & 0.8866 & 0.8657 \\
		AD & 0.6464 & 0.8940 & 0.8020 & 0.8434 & 0.8011 \\
		BS & 0.5146 & 0.9134 & 0.7122 & 0.9130 & 0.9150 \\
		PI+SC & 0.5136 & 0.9165 & 0.7244 & 0.8866 & 0.8657 \\
		PI+AD & 0.6464 & 0.8940 & 0.8020 & 0.8434 & 0.8011 \\
		PI+BS & 0.5146 & 0.9134 & 0.7122 & 0.9130 & 0.9150 \\
		SC+AD & 0.6501 & 0.8974 & 0.8052 & 0.8455 & 0.8050 \\
		SC+BS & 0.5097 & 0.9180 & 0.7152 & 0.8874 & 0.8667 \\
		AD+BS & 0.6363 & 0.8939 & 0.7973 & 0.8399 & 0.7976 \\
		PI+SC+AD & 0.6501 & 0.8974 & 0.8052 & 0.8455 & 0.8050 \\
		PI+SC+BS & 0.5097 & 0.9180 & 0.7152 & 0.8874 & 0.8667 \\
		PI+AD+BS & 0.6363 & 0.8939 & 0.7973 & 0.8399 & 0.7976 \\
		SC+AD+BS & 0.6404 & 0.8979 & 0.8074 & 0.8423 & 0.8025 \\
		FULL & 0.6404 & 0.8979 & 0.8074 & 0.8423 & 0.8025 \\
		\bottomrule
	\end{tabular}
\end{table}

\begin{table}[H]
	\centering
	\caption{Feature effect comparison (correlated valuation)}
	\begin{tabular}{lccccc}
		\toprule
		\bfseries Configuration & \bfseries Steals/Game & \bfseries Chain length & \bfseries Seat 1 & \bfseries Seat 2 & \bfseries Seat 29 \\
		\midrule
		BASE & 87.572 & 4.542 & 1.000 & 0.673 & 0.907 \\
		PI & 87.572 & 4.542 & 1.000 & 0.673 & 0.907 \\
		SC & 52.110 & 2.716 & 0.969 & 0.595 & 0.907 \\
		AD & 65.141 & 3.819 & 1.000 & 0.648 & 0.899 \\
		BS & 97.418 & 5.052 & 1.000 & 0.699 & 0.863 \\
		PI+SC & 52.110 & 2.716 & 0.969 & 0.595 & 0.907 \\
		PI+AD & 65.141 & 3.819 & 1.000 & 0.648 & 0.899 \\
		PI+BS & 97.418 & 5.052 & 1.000 & 0.699 & 0.863 \\
		SC+AD & 61.208 & 3.598 & 0.951 & 0.651 & 0.897 \\
		SC+BS & 53.436 & 2.778 & 0.966 & 0.606 & 0.863 \\
		AD+BS & 67.018 & 3.933 & 1.000 & 0.670 & 0.853 \\
		PI+SC+AD & 61.208 & 3.598 & 0.951 & 0.651 & 0.897 \\
		PI+SC+BS & 53.436 & 2.778 & 0.966 & 0.606 & 0.863 \\
		PI+AD+BS & 67.018 & 3.933 & 1.000 & 0.670 & 0.853 \\
		SC+AD+BS & 62.242 & 3.649 & 0.931 & 0.664 & 0.849 \\
		FULL & 62.242 & 3.649 & 0.931 & 0.664 & 0.849 \\
		\bottomrule
	\end{tabular}
\end{table}

\begin{table}[H]
	\centering
	\caption{Strategy performance by configuration (correlated valuation)}
	\begin{tabular}{lccccc}
		\toprule
		\bfseries Configuration & \texttt{always\_open} & \texttt{always\_steal} & \texttt{coin\_flip} & \texttt{mean\_based} & \texttt{threshold} \\
		\midrule
		BASE & 0.5768 & 0.9162 & 0.6935 & 0.9261 & 0.9260 \\
		PI & 0.5768 & 0.9162 & 0.6935 & 0.9261 & 0.9260 \\
		SC & 0.5978 & 0.9255 & 0.7177 & 0.8129 & 0.8406 \\
		AD & 0.6561 & 0.8935 & 0.7878 & 0.8303 & 0.7884 \\
		BS & 0.5425 & 0.9327 & 0.6602 & 0.9425 & 0.9412 \\
		PI+SC & 0.5978 & 0.9255 & 0.7177 & 0.8129 & 0.8406 \\
		PI+AD & 0.6561 & 0.8935 & 0.7878 & 0.8303 & 0.7884 \\
		PI+BS & 0.5425 & 0.9327 & 0.6602 & 0.9425 & 0.9412 \\
		SC+AD & 0.6571 & 0.8942 & 0.7858 & 0.8296 & 0.7876 \\
		SC+BS & 0.5782 & 0.9444 & 0.6992 & 0.8076 & 0.8438 \\
		AD+BS & 0.6297 & 0.9043 & 0.7776 & 0.8293 & 0.7792 \\
		PI+SC+AD & 0.6571 & 0.8942 & 0.7858 & 0.8296 & 0.7876 \\
		PI+SC+BS & 0.5782 & 0.9444 & 0.6992 & 0.8076 & 0.8438 \\
		PI+AD+BS & 0.6297 & 0.9043 & 0.7776 & 0.8293 & 0.7792 \\
		SC+AD+BS & 0.6348 & 0.9022 & 0.7766 & 0.8288 & 0.7753 \\
		FULL & 0.6348 & 0.9022 & 0.7766 & 0.8288 & 0.7753 \\
		\bottomrule
	\end{tabular}
\end{table}

\begin{table}[H]
	\centering
	\caption{Feature effect comparison (negative correlated valuation)}
	\begin{tabular}{lccccc}
		\toprule
		\bfseries Configuration & \bfseries Steals/Game & \bfseries Chain length & \bfseries Seat 1 & \bfseries Seat 2 & \bfseries Seat 29 \\
		\midrule
		BASE & 61.371 & 3.217 & 0.995 & 0.742 & 0.840 \\
		PI & 61.371 & 3.217 & 0.995 & 0.742 & 0.840 \\
		SC & 48.708 & 2.570 & 0.957 & 0.694 & 0.847 \\
		AD & 55.731 & 3.280 & 0.995 & 0.745 & 0.831 \\
		BS & 61.292 & 3.219 & 0.995 & 0.737 & 0.836 \\
		PI+SC & 48.708 & 2.570 & 0.957 & 0.694 & 0.847 \\
		PI+AD & 55.731 & 3.280 & 0.995 & 0.745 & 0.831 \\
		PI+BS & 61.292 & 3.219 & 0.995 & 0.737 & 0.836 \\
		SC+AD & 54.734 & 3.211 & 0.952 & 0.743 & 0.837 \\
		SC+BS & 48.446 & 2.559 & 0.961 & 0.673 & 0.839 \\
		AD+BS & 56.089 & 3.299 & 0.996 & 0.736 & 0.831 \\
		PI+SC+AD & 54.734 & 3.211 & 0.952 & 0.743 & 0.837 \\
		PI+SC+BS & 48.446 & 2.559 & 0.961 & 0.673 & 0.839 \\
		PI+AD+BS & 56.089 & 3.299 & 0.996 & 0.736 & 0.831 \\
		SC+AD+BS & 54.598 & 3.205 & 0.948 & 0.750 & 0.834 \\
		FULL & 54.598 & 3.205 & 0.948 & 0.750 & 0.834 \\
		\bottomrule
	\end{tabular}
\end{table}

\begin{table}[H]
	\centering
	\caption{Strategy performance by configuration (negative correlated valuation)}
	\begin{tabular}{lccccc}
		\toprule
		\bfseries Configuration & \texttt{always\_open} & \texttt{always\_steal} & \texttt{coin\_flip} & \texttt{mean\_based} & \texttt{threshold} \\
		\midrule
		BASE & 0.5250 & 0.9135 & 0.7035 & 0.9154 & 0.9154 \\
		PI & 0.5250 & 0.9135 & 0.7035 & 0.9154 & 0.9154 \\
		SC & 0.5223 & 0.9158 & 0.7051 & 0.8866 & 0.8594 \\
		AD & 0.6301 & 0.8904 & 0.7900 & 0.8310 & 0.7894 \\
		BS & 0.5197 & 0.9122 & 0.6965 & 0.9101 & 0.9147 \\
		PI+SC & 0.5223 & 0.9158 & 0.7051 & 0.8866 & 0.8594 \\
		PI+AD & 0.6301 & 0.8904 & 0.7900 & 0.8310 & 0.7894 \\
		PI+BS & 0.5197 & 0.9122 & 0.6965 & 0.9101 & 0.9147 \\
		SC+AD & 0.6319 & 0.8941 & 0.7917 & 0.8374 & 0.7905 \\
		SC+BS & 0.5217 & 0.9130 & 0.7011 & 0.8815 & 0.8525 \\
		AD+BS & 0.6269 & 0.8890 & 0.7820 & 0.8290 & 0.7866 \\
		PI+SC+AD & 0.6319 & 0.8941 & 0.7917 & 0.8374 & 0.7905 \\
		PI+SC+BS & 0.5217 & 0.9130 & 0.7011 & 0.8815 & 0.8525 \\
		PI+AD+BS & 0.6269 & 0.8890 & 0.7820 & 0.8290 & 0.7866 \\
		SC+AD+BS & 0.6288 & 0.8890 & 0.7857 & 0.8321 & 0.7848 \\
		FULL & 0.6288 & 0.8890 & 0.7857 & 0.8321 & 0.7848 \\
		\bottomrule
	\end{tabular}
\end{table}

\end{document}